\documentclass[letter,11pt]{article}
\usepackage[margin=1in]{geometry}
\usepackage[utf8]{inputenc}
\usepackage{xspace}
\usepackage[ruled,vlined,linesnumbered]{algorithm2e}
\usepackage{setspace}
\usepackage{braket}
\usepackage{mathtools}
\usepackage{amsmath}
\usepackage{amssymb}
\usepackage{amsfonts}
\usepackage{amsthm}
\usepackage{thmtools,thm-restate}
\usepackage{comment}
\usepackage{multicol}
\usepackage{multirow}
\usepackage[toc,page]{appendix}
\usepackage[textsize=tiny]{todonotes}
\usepackage[inline, shortlabels]{enumitem}

\usepackage[font=scriptsize]{caption}
\usepackage[font=scriptsize]{subcaption}

\usepackage[T1]{fontenc}
\usepackage{textcomp}
\usepackage[utf8]{inputenc}

\usepackage{csquotes}

\usepackage[colorlinks=true]{hyperref}
\usepackage[capitalize,nameinlink,noabbrev]{cleveref}

\renewcommand{\paragraph}[1]{\subparagraph{#1.}}

\newcommand*{\N}{\mathbb{N}}

\renewcommand{\leq}{\leqslant}

\renewcommand{\geq}{\geqslant}

\usepackage{soul}

\theoremstyle{plain}
\newtheorem{theorem}{Theorem}

\newtheorem{corollary}[theorem]{Corollary}
\newtheorem{lemma}[theorem]{Lemma}
\newtheorem{proposition}[theorem]{Proposition}

\theoremstyle{definition}
\newtheorem{definition}[theorem]{Definition}

\theoremstyle{remark}

\newtheorem{observation}{Observation}

\crefname{claim}{Claim}{Claims}
\Crefname{claim}{Claim}{Claims}

\DeclareMathOperator{\opt}{{\sf OPT}}

\newcommand{\vsep}{{\sf VS}}
\newcommand{\is}{{\sf IS}}

\newcommand{\ir}{{\sf IR}}
\newcommand{\vs}{{\sf VS}}

\newcommand{\problemname}{{\sc MinGMConn}}
\newcommand{\sset}{{\mathcal S}}
\newcommand{\iset}{{\mathcal I}}
\newcommand{\rset}{{\mathcal R}}

\DeclarePairedDelimiter{\abs}{\lvert}{\rvert}

\DeclarePairedDelimiter{\intcc}{[}{]}

\newcommand*\cupdot{\mathbin{\mathaccent\cdot\cup}}

\def\mconnected/{M-connected}
\def\mconnectedlong/{Manhattan-connected}

\def\mpath/{M-path}
\def\mpaths/{M-paths}
\def\mpathlong/{Manhattan-path}
\def\mpathslong/{Manhattan-paths}

\def\BinST/{BST}
\def\BinSTlong/{\textsc{Binary Search Tree}}

\def\iiiSAT/{$3$-SAT}

\def\MGMC/{\problemname}
\def\MGMClong/{\textsc{Minimum Generalized Manhattan Connections}}

\DeclareMathOperator{\BST}{\textsf{BST}}

\newcommand{\customfootnotetext}[2]{{
		\renewcommand{\thefootnote}{#1}
		\footnotetext[0]{#2}}}

\begin{document}

\title{On Minimum Generalized Manhattan Connections}

\author{Antonios Antoniadis\thanks{University of Cologne. Work done while the author was at Saarland University \& Max-Planck Institute for Informatics, Germany, and supported by DFG grant AN 1262/1-1. \texttt{antoniadis@cs.uni-koeln.de}.}, Margarita Capretto\thanks{Universidad Nacional de Rosario, Argentina. \texttt{margyc94@gmail.com}.}, Parinya Chalermsook\textsuperscript{1}\thanks{Aalto University, Finland . \texttt{\{parinya.chalermsook,nidia.obscuraacosta,joachim.spoerhase\}@aalto.fi.}},\\Christoph Damerius\thanks{University of Hamburg, Germany. \texttt{\{christoph.damerius,peter.kling\}@uni-hamburg.de.}}, Peter Kling\footnotemark[4], Lukas Nölke\thanks{University of Bremen, Germany. \texttt{noelke@uni-bremen.de}.},\\ Nidia Obscura Acosta\footnotemark[3], Joachim Spoerhase\textsuperscript{2}\footnotemark[3]}

\date{\today}

\maketitle

\customfootnotetext{1}{Currently supported by European Research Council (ERC) under the European Union's Horizon 2020 research and innovation programme (grant agreement No. 759557) and by Academy of Finland Research Fellowship, under grant number 310415.}
\customfootnotetext{2}{Funded by European Research Council (ERC) under the European Union's Horizon 2020 research and innovation programme (grant agreement No. 759557).}

\begin{abstract}
We consider minimum-cardinality Manhattan connected sets with arbitrary demands: Given a collection of points $P$ in the plane, together with a subset of pairs of points in $P$ (which we call \textit{demands}), find a minimum-cardinality superset of $P$ such that every demand pair is connected by a path whose length is the $\ell_1$-distance of the pair.
This problem is a variant of three well-studied problems that have arisen in computational geometry, data structures, and network design: (i)~It is a node-cost variant of the classical Manhattan network problem, (ii)~it is an extension of the binary search tree problem to arbitrary demands, and (iii)~it is a special case of the directed Steiner forest problem.
Since the problem inherits basic structural properties from the context of binary search trees, an $O(\log n)$-approximation is trivial.
We show that the problem is NP-hard and  present an $O(\sqrt{\log n})$-approximation algorithm. Moreover, we provide an $O(\log\log n)$-approximation algorithm for complete $k$-partite 
demands as well as improved results for unit-disk demands and several generalizations.
Our results crucially rely on a new lower bound on the optimal cost that could potentially be useful in the context of BSTs.
\end{abstract}

\section{Introduction}

Given a collection of points $P \subseteq {\mathbb R}^2$ on the plane, the \textit{Manhattan Graph} $G_P$ of $P$ is an undirected graph with vertex set  $V(G_P)= P$ and arcs $E(G_P)$ that connect any vertically- or horizontally-aligned points.
Point $p$ is said to be \emph{\mconnectedlong/ (\mconnected/)} to point~$q$ if $G_P$ contains a shortest rectilinear path from $p$ to $q$ (i.e.\ a path of length~$||p-q||_1$).
In this paper, we initiate the study of the following problem: Given points $P \subseteq {\mathbb R}^2$ and demands~$D \subseteq P \times P$, we want to find a smallest set $Q \subseteq {\mathbb R}^2$ such that every pair of vertices in $D$ is $M$-connected in $G_{P \cup Q}$.
We call this problem \MGMClong/ (\MGMC/), see Figure~\ref{fig:problem} for an illustration.
Variants of this problem have appeared and received a lot of attention in many areas of theoretical computer science, including data structures, approximation algorithms, and computational geometry.
Below, we briefly discuss them, as well as the implications of our results in those contexts.

\begin{figure}[t]
\begin{subfigure}[t]{0.49\linewidth}
\centering
\includegraphics[width=0.98\linewidth]{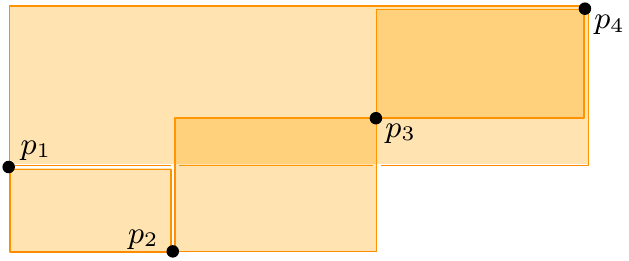}
\end{subfigure}\hfill
\begin{subfigure}[t]{0.49\linewidth}
\centering
\includegraphics[width=0.98\linewidth]{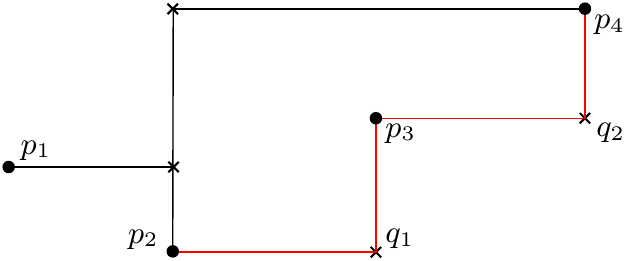}
\end{subfigure}
\caption{
    Left: A Manhattan instance with input points $P=\{p_1,p_2,p_3,p_4\}$, drawn as black disks, and demands $D=\set{(p_1,p_2),(p_1,p_4),(p_2,p_3),(p_3,p_4)}$, drawn as orange rectangles.
    Right: The Manhattan Graph $G_{P \cup Q}$ of a feasible solution $Q$, with points in $Q$ drawn as crosses. 
    Points $p_2$ and $p_4$ are Manhattan-connected via the red path
    $p_2$ - $q_1$ - $p_3$ - $q_2$ - $p_4$. Points $p_1$ and $p_3$ are not Manhattan-connected but also not a demand pair in $D$.
}\label{fig:problem}
\end{figure}

\paragraph{Binary Search Trees (BSTs)} The \textit{Dynamic Optimality Conjecture}~\cite{sleator1985self} is one of the most fundamental open problems in dynamic data structures, postulating the existence of an $O(1)$-competitive binary search tree. Despite continuing efforts and important progress for several decades (see, e.g.,~\cite{demaine2007dynamic,cole2000dynamic,demaine2009geometry,iacono2016weighted,chalermsook2015pattern} and references therein), the conjecture has so far remained elusive, with the best known competitive ratio of $O(\log \log n)$ obtained by Tango trees~\cite{demaine2007dynamic}.
Even in the offline setting, the best known 
algorithm is also a $O(\log \log n)$-approximation; the problem is not even known to be NP-hard.
Demaine, Harmon, Iacono, Kane, and P{\u{a}}tra{\c{s}}cu~\cite{demaine2009geometry} showed that approximating BST is equivalent (up to a constant in the approximation factor) to approximating the node-cost Manhattan problem with ``evolving demand'' (that is, points added to the solution create demands to all existing points).\footnote{In fact, the problem stated in~\cite{demaine2009geometry} is called {\sc MinASS} which appears different from Manhattan problem, but they can be shown to be equivalent. Please see Appendix~\ref{sec: BST and Man} for a detailed discussion on the equivalence.}

The long-standing nature of the $O(\log \log n)$ upper bound could suggest the lower bound answer.
However, the understanding of lower bound techniques for BSTs has been completely lacking: It is not even known whether the problem is NP-hard!
Our work is inspired by the following question.
\begin{quote}
Is it NP-hard to (exactly) compute a minimum cost binary search tree?
\end{quote}
We are, unfortunately, unable to answer this question.
In this paper, we instead present a proof that a natural generalization of the problem  in the geometric view (which is exactly our  \problemname) is NP-hard\footnote{Demaine et al.~\cite{demaine2009geometry} proves an NP-hardness result for \problemname~with uniform demands but allowing the input to contain multiple points on the same row. Their result is incomparable to ours.}.
We believe that our construction and its analysis could be useful in further study of the BST problem from the perspective of lower bounds.

\paragraph{Edge-Cost Manhattan Problem}
Closely related to \problemname~is the edge-cost variant of Manhattan Network~\cite{Gudmundsson}: Given $P \subseteq {\mathbb R}^2$,  our goal is to compute $Q \subseteq {\mathbb R}^2$ such that every pair in $P$ is $M$-connected in $G_{P \cup Q}$, while minimizing the total lengths of the edges used for the connections.
The problem is motivated by various applications in city planning, network layouts, distributed algorithms and VLSI circuit design, and has received attention in the computational geometry community.
Since the edge-cost variant is NP-hard~\cite{chin2011minimum}, the focus has been on approximation algorithms.
Several groups of researchers presented $2$-approximation algorithms~\cite{guo2008fast,chepoi2008rounding}, and this has remained the best known approximation ratio.
Generalizations of the edge-cost variant have been proposed and studied in two directions: In~\cite{das2015approximating}, the authors generalize the Manhattan problem to higher dimension ${\mathbb R}^d$ for $d \geq 2$. The arbitrary-demand case was suggested in~\cite{chepoi2008rounding}. An~$O(\log n)$-approximation algorithm was presented in~\cite{das2018approximating}, which remains the best known ratio.
Our \problemname~problem can be seen as an analogue of~\cite{das2018approximating} in the node-cost setting.
We present an improved approximation ratio of $O(\sqrt{\log n})$, therefore, raising the possibility of similar improvements in the edge-cost variants.

\paragraph{Directed Steiner Forests (DSF)}
\problemname~is a special case of \textit{node-cost directed Steiner forest (DSF)}: Given a directed graph $G=(V,E)$ and pairs of terminals ${\mathcal D} \subseteq V \times V$, find a minimum cardinality subset $S \subseteq V$ such that $G[S]$ contains a path from $s$ to~$t$ for all $(s,t) \in {\mathcal D}$.
DSF is known to be highly intractable, with hardness $2^{\log^{1-\epsilon} |V|}$ unless ${\sf NP} \subseteq {\sf DTIME}(n^{poly \log n})$~\cite{dodis1999design}. The best known approximation ratios are slightly sub-linear~\cite{chekuri2011set,feldman2012improved}. 
Manhattan problems can be thought of as natural, tractable special cases of DSF, with approximability between constant and logarithmic regimes.
For more details, see~\cite{das2015approximating}.

\subsection{Our Contributions}

In this paper, we present both hardness and algorithmic results for \problemname.
\begin{theorem}%
\label{thm:nphard}
The \problemname~problem is NP-hard, even for an input that contains at most one point per row and column. 
\end{theorem}

This result can be thought of as a first step towards developing structural understanding of Manhattan connectivity with respect to lower bounds. 
We believe that such an understanding would come in handy in future study of binary search trees in the geometric view.

Next, we present algorithmic results. Due to the BST structures, an $O(\log n)$-approximation is trivial. The main ingredient in obtaining a sub-logarithmic approximation is an approximation algorithm for the case of ``few'' $x$-coordinates. 
More formally, we say that an input instance is \textit{$s$-thin} if the points in $P$ lie on at most $s$ different $x$-coordinates. 

\begin{theorem}%
\label{thm:xcoord_approx}
There exists an efficient $O(\log s)$-approximation algorithm for an $s$-thin instance~$(P,D)$ of \problemname.
\end{theorem}

In fact, our algorithm produces solutions with $O(\log s \cdot \is(P,D))$ many points, where $\is(P,D) \leq \opt(P,D)$ is the cardinality of a \emph{boundary independent set} (a notion introduced below).
This theorem is tight up to a constant factor, since there exists an input instance $(P,D)$ on $s$ different columns such that the $\opt(P,D) = \Omega(\is(P,D) \log s)$; see Appendix~\ref{sec:tightness of is}.

This theorem, along with the boundary independent set analysis, turns out to be an important building block for our approximation  result, which achieves an approximation ratio that is sublogarithmic in $n$.
\begin{theorem}%
\label{thm:general_approx}
There is an efficient $O(\sqrt{\log n})$-approximation algorithm for \problemname.
\end{theorem}
This gives an improvement over the trivial $O(\log n)$-approximation and may grant some new hope with regards to improving the $O(\log n)$-approximation for the edge-cost variants. 

We provide improved approximation ratios for settings when the graph formed by the demands has a special structure.
\begin{theorem}\label{thm:complete-bipartite}
There is an $O(\log\log n)$-approximation algorithm for \MGMC/ when the demands form a complete $k$-partite graph. 
\end{theorem}
Another set of such results concerns settings where the demand graph is derived from geometry.
Here, points~$p,q$ form a demand if they are within a certain distance $r$ (possibly dependent on $p$ and $q$) of each other.
\begin{theorem}%
\label{thm:disk_approx}
For unit-disk demands, \problemname~admits an $O(1)$-approximation. For two-disk demands, we can achieve an $O(\log\log n)$-approximation.
For (general) disk demands with maximal radius ratio $\Delta$, there is an $O(\log \Delta)$-approximation.
\end{theorem}

\subsection{Overview of Techniques}

The NP-hardness proof is based on a reduction to \iiiSAT/.
In contrast to the uniform case of \problemname, the non-uniform case allows us to encode the structure of a \iiiSAT/ formula in a geometrical manner: we can use demand rectangles to form certain \enquote{paths} (see \cref{fig:construction:overview}).
We exploit this observation in the reduction design by translating clauses and variables into \emph{gadgets}, rectangular areas with specific placement of input points and demands (see \cref{fig:gadgets}).
Variable gadgets are placed between clause gadgets and a dedicated \emph{starting point}.
The crux is to design the instance such that a natural solution to the intra- and inter-gadget demands connects the starting point to either the positive or the negative part of each variable gadget. And, the \mpaths/ leaving a variable gadget from that part can all reach \emph{only} clauses with a \emph{positive} appearance or \emph{only} clauses with a \emph{negative} appearance of that variable respectively.
We refer to such solutions as \emph{boolean solutions}, as they naturally correspond to a variable assignment.
Additional demands between the starting point and the clause gadgets are satisfied by a boolean solution if and only if it corresponds to a satisfying variable assignment.
The main part of the proof is to show that any small-enough solution is a boolean solution.

In the study of any optimization (in particular, minimization) problem, one of the main difficulties is to come up with a strong lower bound on the cost of an optimal solution that can be leveraged by algorithms.
For binary search trees, many such bounds were known, and the strongest known lower bound is called an {\em independent rectangle bound (IR)}.
However, IR is provably too weak for the purpose of \problemname, that is, the gap between the optimal and IR can be as large as~$\Omega(n)$.
We propose to use a new bound, which we call {\em vertically separable demands (VS)}.
This bound turns out to be relatively tight and plays an important role in both our hardness and algorithmic results.
In the hardness result, we use our VS bound to argue about the cost of the optimal in the soundness case.

Our $O(\sqrt{\log n})$-approximation follows the high-level idea of~\cite{chalermsook2019pinning}, which presents a geometric $O(\log \log n)$-approximation for BST. 
Roughly speaking, it argues (implicitly) that two combinatorial properties, which we refer to as 
(A) and~(B), are sufficient for the existence of an $O(\log \log n)$-approximation: (A) the lower bound function is ``subadditive'' with respect to a certain instance partitioning, and (B) the instance is ``sparse'' in the sense that for any input $(P,D)$, there exists an equivalent input~$(P',D')$ such that $|P'| = O(\opt(P,D))$.
In the context of BST, 
(A) holds for the Wilber bound and 
(B) is almost trivial to show.

In the \problemname~problem, we prove that Property (A) holds for the new VS bound. 
However, proving Property (B) seems to be very challenging. We instead show a corollary of Property~(B): There is an $O(\log s)$-approximation algorithm for \problemname,~where $s$ is the number of columns containing at least one input point.
The proof of this relaxed property 
is the main new ingredient of our algorithmic result and is stated in Theorem~\ref{thm:xcoord_approx}.
Finally, we argue that this weaker property still suffices for an $O(\sqrt{\log n})$-approximation algorithm.
For completeness, in Sections~\ref{sec:complete-bipartite} and~\ref{sec:disc_demands}, we discuss special cases (see Theorems~\ref{thm:complete-bipartite} and~\ref{thm:disk_approx}) where we prove that Property (B) holds and thus an $O(\log \log n)$-approximation exists.

\subsection{Outlook and Open Problems}

Inspired by the study of structural properties of Manhattan connected sets and potential applications in BSTs, we initiate the study of \problemname~by proving NP-hardness and giving several algorithmic results.

There are multiple interesting open problems. 
First, can we show that the BST problem is NP-hard? We hope that our construction and analysis using the new VS bound would be useful for this purpose.
Another interesting open problem is to obtain a $o(\log n)$-approximation for the edge-cost variant of the generalized Manhattan network problem.

Finally, it can be shown that our VS bound is sandwiched between OPT and IR. It is an interesting question to study the tightness of the VS bound when estimating the value of an optimal solution. 
Can we show that VS is within a constant factor from the optimal cost of BST? Can we approximate the value of VS efficiently within a constant factor?

\section{Model \& Preliminaries}

Let $P \subseteq {\mathbb R}^2$ be a set of points on the plane.
We say that points $p,q \in P$ are \textit{\mconnectedlong/ (\mconnected/)} in $P$ if there is a sequence of points $p= x_0, x_1,\ldots, x_k = q$ such that
\begin{enumerate*}[(i)]
\item the points~$x_i$ and $x_{i+1}$ are horizontally or vertically aligned for $i = 0,\ldots, k-1$, and
\item the total length satisfies~$\sum_{i=0}^{k-1} ||x_i - x_{i+1}||_1 = ||p-q||_1$.
\end{enumerate*}

In the \emph{minimum generalized Manhattan connections} (\problemname) \emph{problem}, we are given a set of \emph{input points} $P$ and their placement in a rectangular grid with integer coordinates such that there are no two points in the same row or in the same column.
Additionally, we are given a set~$D \subseteq \set{(p, q) | p, q \in P}$ of \emph{demands}.
The goal is to find a set of points $Q$ of minimum cardinality such that $p$ and $q$ are \mconnected/ with respect to $P \cup Q$ for all $(p, q) \in D$.
Denote by $\opt(P,D)$ the size of such a point set.
We differentiate between the points of $P$ and $Q$ by calling them \emph{input points} and \emph{auxiliary points}, respectively.
Since being \mconnected/ is a symmetrical relation, we typically assume $x(p) < x(q)$ for all $(p, q) \in D$.
Here, $x(p)$ and $y(p)$ denote the $x$- and $y$-coordinate of a point $p$, respectively.
In our analysis, we sometimes use the notations~$\intcc{n} \coloneqq \set{1, 2, \dots, n}$ and~$\intcc{n}_0 \coloneqq \intcc{n} \cup \set{0}$, where $n \in \N$.

\paragraph{Connection to Binary Search Trees}
In the \emph{uniform} case where all pairs of input points are connected by a demand, i.e.~$D = \set{(p, q) | p, q \in P }$, this problem is intimately connected to the \BinSTlong/ (\BinST/) problem in the geometric model~\cite{demaine2009geometry}.
Here, we are given a point set $P$ and the goal is to compute a minimum set $Q$ such that \emph{every pair} in $P \cup Q$ is \mconnected/ in $P \cup Q$.
Denote by $\BST(P)$ the optimal value of the \BinST/ problem.

\paragraph{Independent Rectangles and Vertically Separable Demands}
Following Demaine et al.~\cite{demaine2009geometry}, we define the independent rectangle number which is 
a lower bound on $\opt(P,D)$.
For a demand $(p,q) \in D$, denote by $R(p,q)$ the (unique) axis-aligned closed rectangle that has~$p$ and~$q$ as two of its corners.
We call it the \emph{demand rectangle} corresponding to $(p,q)$.
Two rectangles~$R(p,q),R(p',q')$ are called \emph{non-conflicting} if none contains a corner of the other in its interior.
We say a subset of demands $D' \subseteq D$ is \emph{independent}, if all pairs of rectangles in $D'$ are non-conflicting.
Denote by $\ir(P,D)$ the maximum integer $k$ such that there is an independent subset $D'$ of size $k$. We refer to $k$ as the {\em independent rectangle number.}

For uniform demands, the problem admits a $2$-approximation. 
Here, the independent rectangle number plays a crucial role.
Specifically, it was argued in Harmon's PhD thesis~\cite{harmon2006new} that a natural greedy algorithm  costs at most the independent rectangle number and thus yields a $2$-approximation.
In our generalized demand case, however, the independent rectangle number turns out to be a bad estimate on the value of an optimal solution.
Instead, we consider the notion of vertically separable demands, used implicitly in~\cite{demaine2009geometry}.

We say that a subset of demands $D' \subseteq D$ is {\em vertically separable} if there exists an ordering $R_1,R_2,\ldots, R_k$ of its demand rectangles
and vertical line segments $\ell_1,\ell_2\ldots,\ell_k$ such that $\ell_i$ connects the respective interiors of top and bottom boundaries of $R_i$ and does not intersect any $R_j$, for $j>i$.
For an input $(P,D)$, denote by $\vs(P,D)$ the maximum cardinality of such a subset. 
We call a set of demands $D$ \emph{monotone}, if 
either $y(p)<y(q)$ for all $(p,q) \in D$ or $y(p)>y(q)$ for all $(p,q) \in D$.
We assume the former case holds as both are symmetrical.
In the following, we argue that $\vs$ is indeed a lower bound on $\opt$ (proof in Appendix~\ref{sec:proof of new lower bound}).

\begin{lemma}[\cite{demaine2009geometry}]
\label{lem:opt_vs_ir}
Let $(P,D)$ be an input for \problemname.
If $D$ is monotone, then we have $\ir(P,D) \leq \vs(P,D) \leq \opt(P,D)$.
Thus, in general, $\frac{1}{2}\ir(P,D)\leq \vs(P,D) \leq 2\cdot\opt(P,D)$.
\end{lemma}
The charging scheme described in the proof of \cref{lem:opt_vs_ir} injectively maps a demand rectangle~$R$ to a point of the optimal solution that lies in $R$.
This implies the following corollary.
\begin{corollary}%
	\label{cor:coveringsubsolutions}
	Let $D$ be a vertically separable, monotone set of demands and $Q$ a feasible solution. If $\abs{Q} = \abs{D}$, there is a bijection $c\colon Q \to D$ such that $q \in R(c(q))$ for all $q \in Q$.
	In particular, for~$Q' \subseteq Q$ there are at least $\abs{Q'}$ demands from $D$ that each covers some $q \in Q'$.
\end{corollary}
In general, the independent rectangle number and the maximum size of a vertically separable set are incomparable.
By \cref{lem:opt_vs_ir}, we have $\ir(P,D) \leq 2\cdot\vs(P,D)$. However, $\ir(P,D)$ may 
be smaller than $\vs(P,D)$ up to a factor of $n$. To see this, consider $n$ diagonally shifted copies of a demand, e.g. $R_i = R\bigl((i,i),(i+n,i+n)\bigr)$, for $i=1,\dots,n$. Here, $\ir(P,D)=1$ and $\vs(P,D)=n$. Thus, the concept of vertical separability is more useful as a lower bound.

\section{NP-hardness}%
\label{sec:np_hardness}

In this \lcnamecref{sec:np_hardness}, we show \cref{thm:nphard} by reducing the \MGMC/ problem to \iiiSAT/.
In~\iiiSAT/, we are given a formula $\phi$ consisting of $m$ clauses $C_1, C_2, \dots, C_m$ over $n$ variables $X_1, X_2, \dots, X_n$, each clause consisting of three literals.
The goal is to decide whether $\phi$ is satisfiable.
For our reduction, we construct a \MGMC/ instance $(P_{\phi}, D_{\phi})$ and a positive integer $\alpha = \alpha(\phi)$ such that $(P_{\phi}, D_{\phi})$ has an optimal solution of size $\alpha$ if and only if $\phi$ is satisfiable (\cref{lem:phisat,lem:phinotsat}).
This immediately implies \cref{thm:nphard}.
In the following, we identify a demand $d \in D_{\phi}$ with its demand rectangle $R(d)$.
This allows us to speak, for example, of intersections of demands, corners of demands, or points covered by demands.

Our construction of the \MGMC/ instance $(P_{\phi}, D_{\phi})$ is based on different \emph{gadgets} and their connections among each other.
A gadget can be thought of as a rectangle in the Euclidean plane that contains a specific set of input points and demands between these.
We start with a coarse overview of our construction by describing how gadgets are placed and how they interact (\cref{fig:construction:overview}).
We then give the intuition behind our reduction and describe the detailed inner structure of gadgets afterward (\cref{fig:gadgets}).
Because of space constraints, the actual proof of the NP-hardness is given in \cref{sec:app:np-hardness}.

\paragraph{Overview of the Construction}
For each clause $C_j$, we create a \emph{clause gadget} $GC_j$ and for each variable $X_i$, a \emph{variable gadget} $GX_i$.
Clause gadgets are arranged along a descending diagonal line, so all of $GC_j$ is to the bottom-right of $GC_{j-1}$.
Variable gadgets are arranged in the same manner.
This avoids unwanted interference among different clause and variable gadgets, respectively.
The variable gadgets are placed to the bottom-left of all clause gadgets.

For each positive occurrence of a variable $X_i$ in a clause $C_j$, we place a dedicated \emph{connection point}~$p_{ij}^+ \in P_{\phi}$ as well as suitable \emph{connection demands} from $p_{ij}^+$ to a dedicated inner point of $GX_i$ and to a dedicated inner point of $GC_j$.
Their purpose is to force optimal \MGMC/ solutions to create specific \mpaths/ (going first up and then right in a narrow corridor) connecting a variable to the clauses in which it appears positively.
We call the area covered by these two demands a \emph{(positive) variable-clause path}.
Similarly, there are connection points $p_{ij}^- \in P_{\phi}$ with suitable demands for negative appearances of $X_i$ in $C_j$, creating a \emph{(negative) variable-clause path} (going first right and then up in a narrow corridor).

Finally, there is a \emph{starting point} $S \in P_{\phi}$ to the bottom-left of all other points.
It has a demand to a \emph{clause point} $c_j$ in the top-right of each clause gadget $GC_j$ (an \emph{$SC$ demand}) and to a \emph{variable point} $x_i$ in the bottom-left of each variable gadget $GX_i$ (an \emph{$SX$ demand}).
The inside of clause gadgets simply provides different entrance points for the variable-clause paths, while the inside of variable gadgets forces an optimal solution to choose between using either only positive or only negative variable-clause paths.
We will use these choices inside variable gadgets to identify an optimal solution for $(P_{\phi}, D_{\phi})$ with a variable assignment for $\phi$.
Details about clause and variable gadgets are given in \cref{fig:gadgets} and at the end of this \lcnamecref{sec:np_hardness}.

\begin{figure}
\centering
\includegraphics[width=0.7\linewidth]{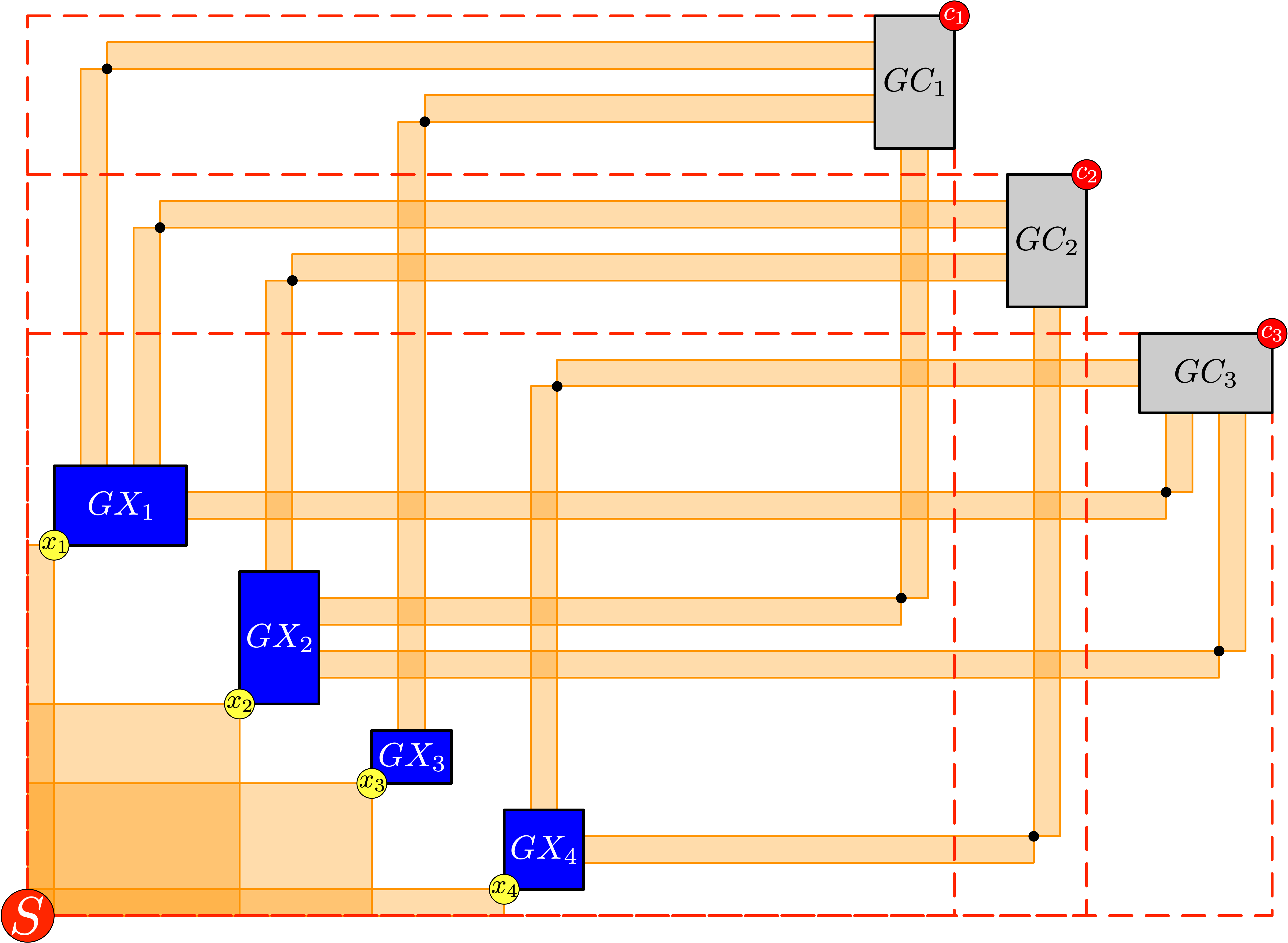}
\caption{%
    \MGMC/ instance $(P_{\phi}, D_{\phi})$ for $\phi = (X_1 \lor \neg X_2 \lor X_3) \land (X_1 \lor X_2 \lor \neg X_4) \land (\neg X_1 \lor \neg X_2 \lor X_4)$.
    Input points are shown as (red, yellow, or black) disks.
    For clause and variable gadgets, we show only the clause points $c_j$ and the variable points~$x_i$; their remaining inner points and demands are illustrated in \cref{fig:gadgets}.
    The small black disks represent the connection points~$p_{ij}^+,p_{ij}^-$.
    Non-$SC$ demands are shown as shaded, orange rectangles, while $SC$ demands are shown as dashed, red rectangles.
}%
\label{fig:construction:overview}
\end{figure}

\paragraph{Intuition of the Reduction}
Our construction is such that the non-$SC$ demands (including the ones within gadgets) form a monotone, vertically separable demand set.
Thus, for
\begin{equation}
\begin{aligned}
           D_{\overline{SC}}
&\coloneqq \set{d \in D_{\phi} | \text{$d$ is not an $SC$ demand}}
&&\quad\text{and}\quad&
   \alpha
&= \alpha(\phi)
\coloneqq \abs{D_{\overline{SC}}}
,
\end{aligned}
\end{equation}
\cref{lem:opt_vs_ir} implies that any solution $Q_{\phi}$ for $(P_{\phi}, D_{\phi})$ has size at least $\alpha$.

The first part of the reduction (\cref{lem:phisat}) shows that if $\phi$ is satisfiable, then there is an (optimal) solution $Q_{\phi}$ of size $\alpha$.
This is proven by constructing a family of \emph{boolean solutions}.
These are (partial) solutions $Q_{\phi}$ that can be identified with a variable assignment for $\phi$ and that have the following properties:
$Q_{\phi}$ has size $\alpha$ and satisfies all non-$SC$ demands.
Additionally, it can satisfy an $SC$ demand $(S, c_j)$ only by going through some variable $x_i$, where such a path exists if and only if~$C_j$ is satisfied by the value assigned to $X_i$ by (the variable assignment) $Q_{\phi}$.
In particular, if $\phi$ is satisfiable, there is a boolean solution $Q_{\phi}$ satisfying all $SC$ demands.
This implies that $Q_{\phi}$ is a solution to $(P_{\phi}, D_{\phi})$ of (optimal) size~$\alpha$.

\Cref{lem:phinotsat} provides the other direction of the reduction, stating that if there is a solution~$Q_{\phi}$ for $(P_{\phi},D_{\phi})$ of size $\alpha$, then $\phi$ is satisfiable.
Its proof is more involved and is made possible by careful placement of gadgets, connection points, and demands.
In a first step, we show that the small size of $Q_{\phi}$ implies that different parts of our construction each must be satisfied by only a few, dedicated points from $Q_{\phi}$.
For example, $Q_{\phi}$ has to use exactly $n$ points to satisfy the $n$ $SX$ demands $(S, x_i)$.
Another result (\cref{prop:triangular_demands}) about \enquote{triangular} instances (e.g., the triangular grid formed by the $n$ $SX$ demands, see \cref{fig:construction:overview}) states that, here, optimal solutions must lie on grid lines inside the \enquote{triangle}.
We conclude that any \mpath/ from $S$ to a clause point $c_j$ must go through exactly one variable point $x_i$.
Similarly, we show that the $6m$ connection demands (forming the $3m$ variable-clause paths) are satisfied by $6m$ points from $Q_{\phi}$ and, since they are so few, each of these points lies in the corner of a connection demand.
This ensures that \mpaths/ cannot cheat by, e.g., \enquote{jumping} between different variable-clause paths.
More precisely, such a path can be entered only at the variable gadget where it starts and be left only at the clause gadget where it ends.

All that remains to show is that there cannot be two \mpaths/ entering a variable gadget~$GX_i$ (which they must do via $x_i$) such that one leaves through a positive and the other through a negative variable-clause path.
We can then interpret $Q_{\phi}$ as a boolean solution (the variable assignment for~$X_i$ being determined by whether \mpaths/ leave $GX_i$ through positive or through negative variable-clause paths).
Since $Q_{\phi}$ satisfies all demands, in particular all $SC$ demands, the corresponding variable assignment satisfies all clauses.

\paragraph{Details of Clause \& Variable Gadgets}
We recommend to keep \cref{fig:gadgets:clause,fig:gadgets:variable} close at hand when reading the following gadget descriptions.
\begin{figure}
\begin{subfigure}[t]{0.44\linewidth}
\includegraphics[width=0.9\linewidth]{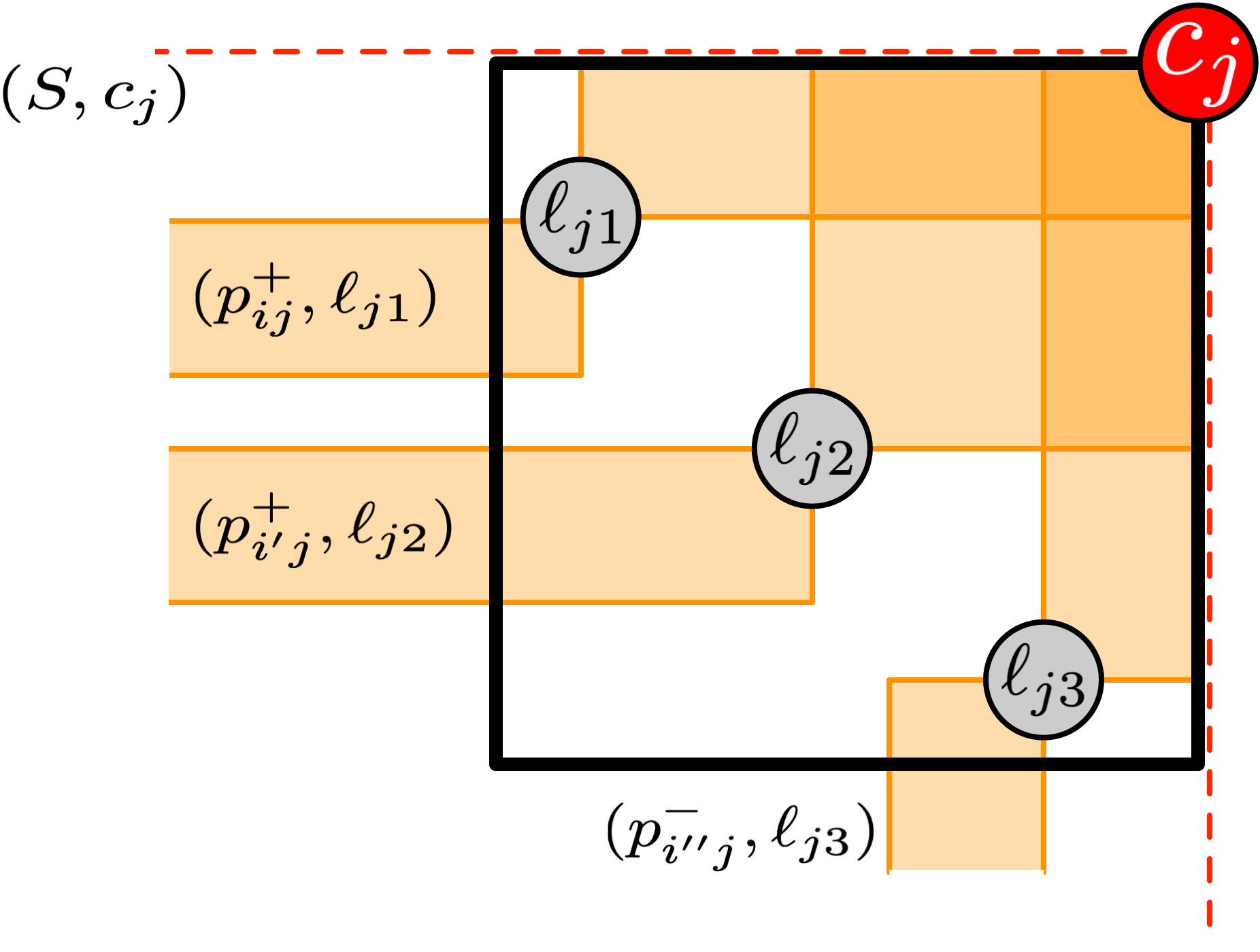}
\subcaption{%
    Clause gadget $GC_j$ for $C_j = X_i \lor X_{i'} \lor \neg X_{i''}$.
}\label{fig:gadgets:clause}
\end{subfigure}
\hfill
\begin{subfigure}[t]{0.52\linewidth}
\includegraphics[width=\linewidth]{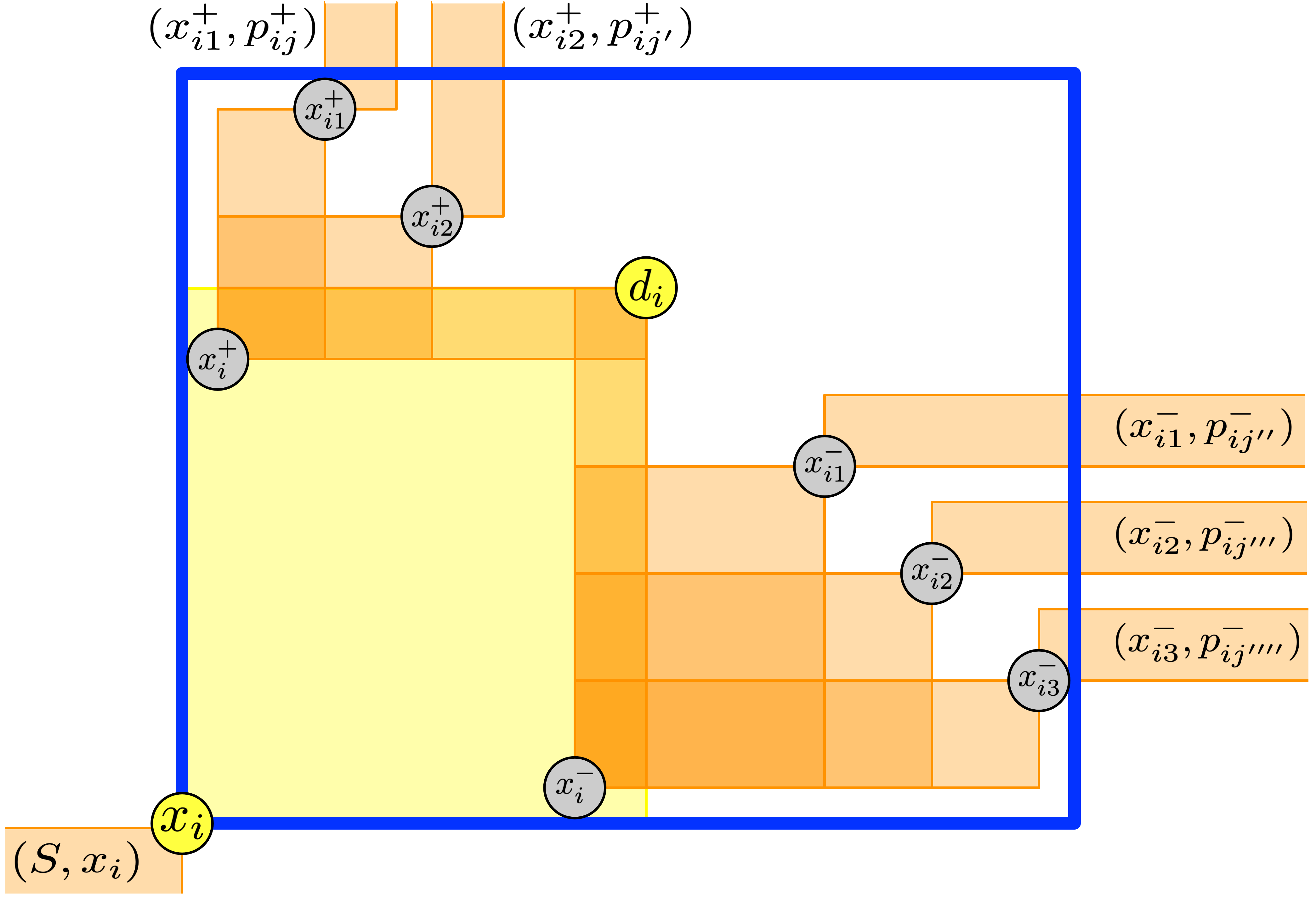}
\subcaption{%
    Variable gadget $GX_i$ for $X_i$ occurring twice positively and thrice negatively.
}\label{fig:gadgets:variable}
\end{subfigure}
\caption{%
    Examples for a clause and a variable gadgets.
    As in \cref{fig:construction:overview}, input points are shown as circles and $SC$~demands are shown as dashed, red rectangles.
    The $XD$ demand $(x_i, d_i)$ is shown as a shaded, yellow rectangle.
    All remaining (non-$SC$ and non-$XD$) demands are again shown as shaded, orange rectangles.
}\label{fig:gadgets}
\end{figure}
The \emph{clause gadget} $GC_j$ for clause $C_j$ contains the \emph{clause point} $c_j$ and three \emph{(clause) literal points} $\ell_{j1}, \ell_{j2}, \ell_{j3}$.
The clause point is in the top-right.
The literal points represent the literals of~$C_j$ and form a descending diagonal within the gadget such that positive are above negative literals.
For each literal point $\ell_{jk}$, there is a demand $(\ell_{jk}, c_j)$.
Moreover, if $\ell_{jk}$ is positive and corresponds to the variable $X_i$, then there is a \emph{(positive) connection demand}~$(p_{ij}^+, \ell_{jk})$.
Similarly, if $\ell_{jk}$ is negative, there is a \emph{(negative) connection demand} $(p_{ij}^-, \ell_{jk})$.
Finally, there is the \emph{$SC$ demand} $(S, c_j)$.

The \emph{variable gadget} $GX_i$ for variable $X_i$ contains
\begin{itemize*}[label=, afterlabel=]
\item the \emph{variable point} $x_i$,
\item two \emph{(variable) literal points} $x_i^+, x_i^-$,
\item one \emph{demand point} $d_i$, as well as
\item $n_i^+$ positive and $n_i^-$ negative \emph{literal connectors}~$x_{ik}^+$ and $x_{ik}^-$, respectively.
\end{itemize*}
Here, $n_i^+$ and $n_i^-$ are from $[m]_0$ and denote the number of positive and negative occurrences of $X_i$ in $\phi$, respectively.
The variable point is in the bottom-left.
The literal connectors and the demand point form a descending diagonal in the top-right, with the positive literal connectors above and the negative literal connectors below the demand point.
The literal points $x_i^+, x_i^-$ lie in the interior of the rectangle spanned by $x_i$ and $d_i$, close to the top-left and bottom-right corner respectively.
They are moved slightly inward to avoid identical $x$- or $y$-coordinates.
Inside the gadgets, we have demands of the form
\begin{itemize*}[label=, afterlabel=]
\item $(x_i^+, x_{ik}^+)$ and $(x_i^-, x_{ik}^-)$ between literal points and literal connectors,
\item $(x_i^+, d_i)$ and $(x_i^-, d_i)$ between literal points and the demand point, as well as
\item $(x_i, d_i)$ between the variable point and the demand point (an \emph{$XD$ demand}).
\end{itemize*}
Towards the outside, we have the \emph{positive/negative connection demands}
\begin{itemize*}[label=, afterlabel=]
\item \smash{$(x_{ik}^+, p_{ij}^+)$} if the $k$-th positive literal of $X_i$ occurs in $C_j$
\item and~$(x_{ik}^-, p_{ij}^-)$ if the $k$-th negative literal of $X_i$ occurs in $C_j$
\end{itemize*}
as well as the \emph{$SX$ demand} $(S, x_i)$.

\section{An Approximation Algorithm for $s$-Thin Instances}%
\label{sec:log_xcoords_approx}

In this section, we present and analyze an aproximation algorithm for $s$-thin instances (where points in $P$ lie on at most $s$ distinct $x$-coordinates). In particular, we allow more than one point to share an $x$-coordinate. However, we still require any two distinct points to have distinct $y$-coordinates. 
We show an approximation ratio of $O(\log s)$, proving Theorem~\ref{thm:xcoord_approx}.

An \emph{$x$-group} is a maximal subset of $P$ having the same $x$-coordinate.
Note that an $O(\log s)$-approximation for $s$-thin instances can be obtained via a natural ``vertical'' divide-and-conquer algorithm that recursively divides the $s$ many~$x$-groups in two subinstances with roughly $s/2$ many~$x$-groups each. (Section~\ref{sec:sublogarithmic_general_approx:algorithm} contains a more general version of this algorithm subdividing into an arbitrary number of subinstances.) The analysis of this algorithm uses the number of input points as a lower bound on $\opt$.
However, such a bound is not sufficient for our purpose of deriving an~$O(\sqrt{\log n})$-approximation.

In this section, we present a different algorithm, based on ``horizontal'' divide-and-conquer (after a pre-processing step to sparsify the set of $y$-coordinates in the input via minimum hitting sets). Using horizontal rather than vertical divide-and-conquer may seem counter-intuitive at first glance as the number of $y$-coordinates in the input is generally unbounded in~$s$. Interestingly enough, we can give a \emph{stronger} guarantee for this algorithm by bounding 
the cost of the approximate solution against what we call a \emph{boundary independent set}. Additionally, we show that the size of a such a set is always upper bounded by the maximum number of vertically separable demands. 
This directly implies Theorem~\ref{thm:xcoord_approx}, since $2\opt$ is an upper bound on the number of vertically separable demands (c.f. Lemma~\ref{lem:opt_vs_ir}).  Even more importantly, our stronger bound allow us to prove Theorem~\ref{thm:general_approx} in the next section since vertically separable demands fulfill the subadditivity property mentioned in the introduction. In the proof of Theorem~\ref{thm:general_approx}, an arbitrary $O(\log s)$-approximation algorithm would not suffice.

By losing a factor~$2$ in the approximation ratio, we may assume that the demands are monotone, since we can handle pairs with $x(p)<x(q)$ and $y(p)>y(q)$ symmetrically.

\begin{definition}[Left \& right demand segments]
Let $(P,D)$ be an input instance. For~$R(p,q)\in Q$, denote by $\lambda(p,q)$ the vertical segment that connects $(x(p),y(p))$ and  $(x(p),y(q))$. Similarly, denote by $\rho(p,q)$ the vertical segment that connects $(x(q),y(p))$ and $(x(q), y(q))$. That is, $\lambda(p,q)$ and $\rho(p,q)$ are simply the left and right boundaries of rectangle $R(p,q)$. 
\end{definition}

\paragraph{Boundary Independent Sets}
A \textit{left boundary independent set} consists of pairwise non-overlapping segments~$\lambda(p,q)$, a \emph{right boundary independent set} of pairwise non-overlapping segments~$\rho(p,q)$. A \emph{boundary independent set} refers to either a left or a right boundary independent set. 
Denote by $\is(P,D)$ the size of a maximum boundary independent set.

The following lemma implies that it suffices to work with boundary independent sets instead of vertical separability. The main advantage of doing so, is that (i) for \is, we do not have to identify any ordering of the demand subset, (ii) one can compute $\is(P,D)$ efficiently, and (iii) we can exploit geometric properties of interval graphs, as we will do below.  

\begin{lemma}\label{lem:bound-indep-sets}
  For any instance $(P,D)$ 
  we have that $\is(P,D)\leq\vsep(P,D)$. Moreover, one can compute a maximum boundary independent set in polynomial time.
\end{lemma}
\begin{proof}
  Let $I \subseteq D$ be a maximum boundary independent set, and let $b$ be the vertical line corresponding to the smallest $x$-coordinate of the input points. Assume $I$ consists of left sides of demand rectangles. Consider all the demands $(p,q)\in I$ where $\lambda(p,q)$ lies on the vertical line $b$. Since the left sides of the corresponding demand rectangles form an independent set we can cut them in any order along $b$ without intersecting any other demand in $I$. We can then remove these demands and proceed recursively to the next smallest $x$-coordinate until all demands have been cut. This completes the proof of the inequality.

Now, we discuss efficient computability. Note that a left (and thus a right or general) maximum boundary independent set can be computed efficiently by determining a maximum independent set of the intervals 
$\lambda(p,q)$ for $(p,q)\in D$ along all $x$-coordinates separately. 
The maximum independent set of intervals can be computed in polynomial time~\cite{Snoeyink07}.
\end{proof}

\paragraph{Algorithm Description}
In Algorithm~\ref{alg:inter-strip-general}, we present algorithm {\sc HorizontalManhattan}, which produces a Manhattan solution of cost $O(\log s)\cdot\is(P,D)$, where $s$ is the number of $x$-groups in $P$. 
The algorithm initially computes  a set of ``crucial rows'' $\rset \subseteq {\mathbb R}$ by computing a minimum hitting set in the interval set~$\iset = \{[y(p), y(q)] \mid (p,q) \in D\}$. 
In particular, the set $\rset$ has the following property. 
For each~$j \in \rset$, let $\ell_j$ be a horizontal line drawn at $y$-coordinate $j$. 
Then the lines $\{\ell_j\}_{j \in \rset}$ stab every rectangle in $\{R(p,q)\}_{(p,q) \in D}$. The following observation follows from the fact that the interval hitting set is equal to the maximum interval independent set.

\begin{observation}
$|\rset| \leq \is(P,D)$ 
\end{observation}

After computing 
$\rset$, the algorithm calls a subroutine  {\sc HorizontalDC} (see Algorithm~\ref{alg:horizontal-manhattan}), which recursively adds points to each such row in a way that guarantees a feasible solution.

\vspace{0.2cm}
\begin{algorithm}[t]
\SetKwInOut{Input}{input}
\SetKwInOut{Output}{output}
\Input{ Instance $(P,D)$ with $s$ distinct $x$-coordinates}
\Output{ Feasible solution to $(P,D)$ of size at most $O(\log s)\cdot\is(P,D)$}
$\mathcal{I}\gets\{\,\textnormal{interval }[y(p),y(q)]\mid (p,q)\in D\,\}$\;
$\rset\gets$minimum hitting set for $\mathcal{I}$\;
\Return{\textnormal{HorizontalDC}$(P,D,\rset)$\;}
\caption{{\sc HorizontalManhattan}$(P,D)$\label{alg:inter-strip-general}}
\end{algorithm}

\begin{algorithm}[t]
\SetKwInOut{Input}{input}
\SetKwInOut{Output}{output}
\Input{ Instance $(P,D)$ with rows $\rset$}
\Output{ Feasible solution to $(P,D)$ computed via horizontal divide-and-conquer}
$Q\gets\emptyset$\;
$m\gets$ median of $\rset$\;
$D_m\gets\{\,(p,q)\in D\mid y(p)\leq m\leq y(q)\,\}$\;
\ForEach{$(p,q)\in D_m$}{
   $Q\gets Q\cup\{(x(p),m), (x(q),m)\}$\label{line:solve-median}
}
$D_t\gets\{\,(p,q)\in D\mid y(p)> m\,\}$\;
$P_t\gets\{\,p,q\mid (p,q)\in D_t\,\}$\;
$\rset_t\gets\{\, r\in \rset\mid r>m\,\}$\;
$Q\gets Q\cup\textnormal{HorizontalDC}(P_t,D_t,\rset_t)$\;
$D_b\gets\{\,(p,q)\in D\mid y(q)<m\,\}$\;
$P_b\gets\{\,p,q\mid (p,q)\in D_b\,\}$\;
$\rset_b\gets\{\, r\in \rset\mid r<m\,\}$\;
$Q\gets Q\cup\textnormal{HorizontalDC}(P_b,D_b,\rset_b)$\;
\Return{$Q$\;}
\caption{{\sc HorizontalDC}$(P,D,\rset)$\label{alg:horizontal-manhattan}}
\end{algorithm}

\paragraph{Analysis} 

\begin{lemma}[Feasibility]
The algorithm {\sc HorizontalManhattan} produces a feasible solution in polynomial time. 
\end{lemma}
\begin{proof}
The feasibility of the computed solution follows easily by induction over the number of rows, using the fact that $\rset$ is a hitting set of $\iset$ and that in the for loop in Algorithm~\ref{alg:horizontal-manhattan} (Line~\ref{line:solve-median}), a feasible solution for $D_m$ is computed. 
It is also easy to see that the running time is polynomial in the input size. The set $\rset$ has size at most~$n$, and thus {\sc HorizontalDC} is called at most $O(\log n)$ times. In each call, at most $s$ points are added, and $s$ is also upper bounded by $n$. Finally, assignments can be done in polynomial time.
\end{proof}

\begin{lemma}[Cost]\label{lem:inter-strip}
  For any $s$-thin instance $(P,D)$, algorithm 
  {\sc HorizontalManhattan} 
  outputs a solution of cost $O(\log s)\cdot\is(P,D)$.
\end{lemma}
\begin{proof}
Let $r=|\rset|$ be the number of rows computed in Algorithm~\ref{alg:inter-strip-general}. There exists a subset $I \subseteq \iset$, also of size $r$, that is an independent set~\cite{Golumbic2004perfect}. 
Define $L = \{\lambda(p,q) \mid [y(p), y(q)] \in I\}$ to be the set of corresponding left sides of the demands in $I$. In particular, the segments in~$L$ are disjoint and~$|L| = r$. 
We upper bound the cost of our solution as follows. For each added point, define a {\em witness interval}, witnessing its cost. 
The total number of points is then roughly bounded by the number of witness intervals, which we show to be $O(\log s) \is(P,D)$.  

We enumerate the recursion levels of Algorithm~\ref{alg:horizontal-manhattan} from $1$ to $\lceil \log r \rceil$ in a top-down fashion in the recursion tree. 
In each recursive call, at most $s$ many points are added to $Q$ in line~\ref{line:solve-median}---one for each distinct $x$-coordinate. 
Hence, during the first $\lceil \log r \rceil - \lceil \log s \rceil$ recursion levels at most~$s\cdot 2^{\lceil \log r \rceil - \lceil \log s \rceil}= O(s\cdot\frac{r}{s})=O(r)$ many points are added to $Q$ in total. We associate each of these points with one unique left side in $L$ in an arbitrary manner. For each of these points, we call its associated left side the  witness of this point.

For any point added to $Q$ in line~\ref{line:solve-median} in one of the last $\lceil \log s \rceil$ recursion levels, pick the first (left or right) side of a demand rectangle that led to including this point.  More precisely, if, in line~\ref{line:solve-median}, we add point $(x(p),m)$ to $Q$ for the first time (which means that this point has not yet been added to~$Q$ via a different demand) then associate $\lambda(p,q)$ as a witness. Analogously, if we add $(x(q),m)$ for the first time then associate $\rho(p,q)$ as a witness.

Overall, we have associated to each point in the final solution a uniquely determined witness, which is a left or a right side of some demand rectangle. Note that any (left or right) side of a rectangle may be assigned as a witness to two solution points (once in the top recursion levels and once in the bottom levels). In such a case we create a duplicate of the respective side and consider them to be distinct witnesses.

Two witnesses added in the last $\lceil \log s \rceil$ recursion levels can intersect only if the recursive calls lie on the same root-to-leaf path in the recursion tree. Otherwise, they are separated by the median row of the lowest common ancestor in the recursion tree and cannot intersect.
With this observation and the fact that the witnesses in $L$ form an independent set, we can bound the maximum clique size in the intersection graph of all witnesses by $1+\lceil \log s \rceil$.

This graph is an interval graph. Since interval graphs are perfect~\cite{Golumbic2004perfect}, there exists a $(1+\lceil \log s \rceil)$-coloring in this graph. Hence, there exists an independent set of witnesses of size $1/(\lceil \log s \rceil+1)$ times the size of the Manhattan solution. Taking all left or all right sides of demands in this independent set (whichever is larger) gives a \emph{boundary} independent set of size at least $1/(2(1+\lceil \log s \rceil))$ times the cost of the Manhattan solution.
\end{proof}

We conclude the section by noting that the proof of Theorem~\ref{thm:xcoord_approx} directly follows  by combining Lemmata~\ref{lem:opt_vs_ir},  \ref{lem:bound-indep-sets} and~\ref{lem:inter-strip}. As mentioned in the introduction, the factor $O(\log s)$ in Lemma~\ref{lem:inter-strip} is tight in the strong sense that there is an \MGMC/ instance $(P,D)$ with $s$ distinct $x$-coordinates such that $\opt(P,D)=\Omega(\is(P,D)\log s)$. See Appendix~\ref{sec:tightness of is}.

\section{A Sublogarithmic Approximation Algorithm}%
\label{sec:sublogarithmic_general_approx}

In this section, we give an overview of how to leverage the $O(\log s)$-approximation for $s$-thin instances to design an $O(\sqrt{\log n})$-approximation algorithm for general instances. 

\paragraph{Sub-instances} Let $(P,D)$ be an instance of \problemname~and let $B$ be a bounding box for~$P$, that is, $P\subseteq B$. Let $\sset=\{S_1,\ldots, S_s\}$ be a collection of $s$ vertical strips, ordered from left to right, that are obtained by drawing $s-1$ vertical lines that partition $B$. We naturally create $s+1$ sub-instances as follows. (See also Figure~\ref{fig:sub-instances}.)
First, we have $s$ \emph{intra-strip instances}~$\{(P_i,D_i)\}_{i \in [s]}$ such that $P_i = P \cap S_i$ and $D_i = D \cap (S_i \times S_i)$. 
Next, we have the \emph{inter-strip instance}~$\pi_{\sset}(D) = (P',D')$ where $P'$ is obtained by collapsing each strip in $\sset$ into a single column and $D'$ is obtained from collapsing demands accordingly. For each point $p \in P$, denote by $\pi_{\sset}(p)$ a copy of $p$ in $P'$ after collapsing. Note that this is a simplified description of the instances that avoids some technicalities. For a precise definition, see Appendix~\ref{sec:sublogarithmic_general_approx:analysis}.

\begin{figure}[tb]
    \centering
    \includegraphics[width=0.62\textwidth]{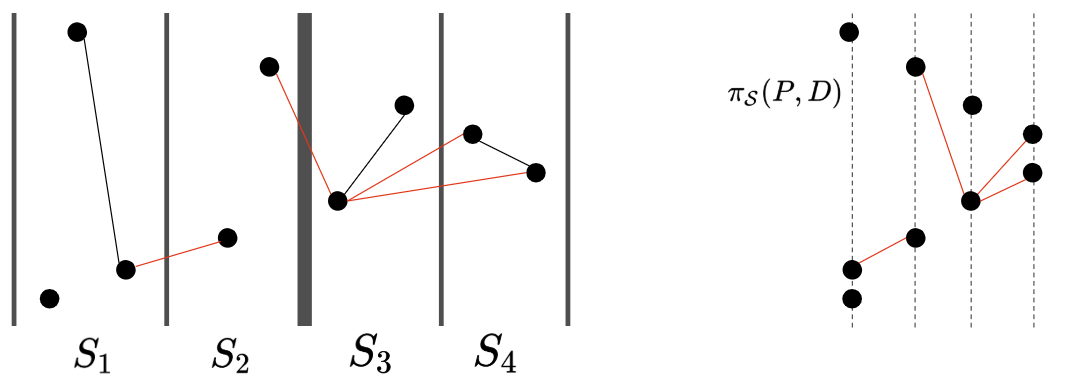}
    \caption{An illustration of the inter-strip instance. Each strip $S_i$ is collapsed into one column. The red demands are demands between pairs of points lying inside different strips. The black demands are demands that are handled by intra-strip instances.}
    \label{fig:sub-instances}
\end{figure}

\paragraph{Sub-additivity of VS}
The following is our sub-additivity property that we use crucially in our divide-and-conquer algorithm.

\begin{restatable}{lemma}{stripSeparability}
\label{lem:strip-separability}
  If $(P,D)$ is an instance of \MGMC/ with strip subdivision~$\mathcal{S}$, then
  \begin{displaymath}
    \vsep(P,D)\geq\vsep(\pi_{\mathcal{S}}(P,D))+\sum_{S\in\mathcal{S}}\vsep(P\cap S, D\cap (S\times S))\,.
  \end{displaymath}
\end{restatable}
\begin{proof}
  Let $(P',D')=\pi_{\mathcal{S}}(P,D)$ be the inter-strip instance.
  Let $D''\subseteq D'$ be a vertically separable demand subset of size $\vsep(P',D')$. Let $D_{\textnormal{I}}=\{(p,q)\in D\mid (\pi_{\mathcal{S}}(p),\pi_{\mathcal{S}}(q))\in D''\,\}$ be the demand subset of $D$ corresponding to $D''$.  For any $S\in\mathcal{S}$, let $D_S$ be a vertically separable demand subset of $D\cap(S\times S)$ of cardinality $\vsep(P\cap S, D\cap(S\times S))$. We claim that $D_{\textnormal{I}}\cup\bigcup_{S\in\mathcal{S}}D_S$ is a vertically separable demand subset of $D$. Obviously, its size is equal to the right side of the inequality that we want to prove.

Since $D''$ is vertically separable, we can sequentially cut it via vertical line segments. W.l.o.g.\, we may assume that these segments lie on strip boundaries. Observe also that any vertical line segment~$z$ on some strip boundary cuts (intersects) some rectangle $R(\pi_{\mathcal{S}}(p),\pi_{\mathcal{S}}(q))$ with~$(\pi_{\mathcal{S}}(p),\pi_{\mathcal{S}}(q))\in D''$ if and only if cuts (intersects) $R(p,q)$ where $(p,q)\in D_{\textnormal{I}}$.  Hence, the same sequence of vertical line segments on the strip boundaries that cuts $D''$ also cuts $D_{\textnormal{I}}$.

After cutting $D_{\textnormal{I}}$, process the strips $S\in\mathcal{S}$ in an arbitrary order and cut the demand set $D_S$ with segments contained in the interior of $S$. Obviously, there is no interference between cutting within different strips and none of the segments we used for cutting~$D_{\textnormal{I}}$ intersects any rectangle in~$D_S$ because they lie on the strip boundaries.
\end{proof}

\paragraph{Divide-and-conquer} Choose the strips $\sset$ so that $s=|\sset| = 2^{\sqrt{\log n}}$. Thus the inter-strip instance admits an approximation of ratio $O(\log s) = O(\sqrt{\log n})$; in fact, we obtain a solution of cost $O(\sqrt{\log n}) \vs(\pi_{\sset}(P,D))$. We recursively solve each intra-strip instance $(P_i, D_i)$, and combine the solutions from these $s+1$ sub-instances. (Details on how the solution can be combined are deferred to Appendix~\ref{sec:sublogarithmic_general_approx:analysis}.) 

We show by induction on the number of points that for any instance $(P,D)$ the cost of the computed solution is $O(\sqrt{\log n}) \vs(P,D)$. (Here, we do not take into account the cost incurred by combining the solutions to the sub-instances.) By induction hypothesis, we have for each $(P_i,D_i)$ a solution of cost $O(\sqrt{\log n}) \vs(P_i,D_i)$ since $|P_i|<|P|$. Note that we cannot use the induction hypothesis for the inter-strip instance since $|P'|=|P|$, which is why we need the $O(\log s)$-approximation algorithm. Using sub-additivity we obtain: 
  \begin{displaymath}
   O(\sqrt{\log n})\left(  \vsep(\pi_{\mathcal{S}}(P,D))+\sum_{S\in\mathcal{S}}\vsep(P\cap S, D\cap (S\times S)) \right) = O(\sqrt{\log n}) \vs(P,D) \,.
  \end{displaymath}
There is an additional cost incurred by combining the solutions of the sub-instances to a feasible solution of the current instance. In  Appendix~\ref{sec:sublogarithmic_general_approx:analysis} we argue that this can be done at a cost of $O(\opt)$ for each of the $\log n/\log s=\sqrt{\log n}$ many levels of the recursion. (This prevents us from further improving the approximation factor by picking $\smash{s=2^{o(\sqrt{\log n})}}$.)

\newpage

\bibliographystyle{plainurl}
\bibliography{refs}

\begin{thebibliography}{10}

\bibitem{chalermsook2019pinning}
Parinya Chalermsook, Julia Chuzhoy, and Thatchaphol Saranurak.
\newblock Pinning down the strong {W}ilber 1 bound for binary search trees.
\newblock {\em arXiv preprint arXiv:1912.02900}, 2019.

\bibitem{chalermsook2015pattern}
Parinya Chalermsook, Mayank Goswami, L{\'a}szl{\'o} Kozma, Kurt Mehlhorn, and
  Thatchaphol Saranurak.
\newblock Pattern-avoiding access in binary search trees.
\newblock In {\em 2015 IEEE 56th Annual Symposium on Foundations of Computer
  Science}, pages 410--423. IEEE, 2015.

\bibitem{chekuri2011set}
Chandra Chekuri, Guy Even, Anupam Gupta, and Danny Segev.
\newblock Set connectivity problems in undirected graphs and the directed
  steiner network problem.
\newblock {\em ACM Transactions on Algorithms (TALG)}, 7(2):1--17, 2011.

\bibitem{chepoi2008rounding}
Victor Chepoi, Karim Nouioua, and Yann Vaxes.
\newblock A rounding algorithm for approximating minimum {M}anhattan networks.
\newblock {\em Theoretical Computer Science}, 390(1):56--69, 2008.

\bibitem{chin2011minimum}
Francis~YL Chin, Zeyu Guo, and He~Sun.
\newblock Minimum {M}anhattan network is np-complete.
\newblock {\em Discrete \& Computational Geometry}, 45(4):701--722, 2011.

\bibitem{cole2000dynamic}
Richard Cole.
\newblock On the dynamic finger conjecture for splay trees. part ii: The proof.
\newblock {\em SIAM Journal on Computing}, 30(1):44--85, 2000.

\bibitem{das2018approximating}
Aparna Das, Krzysztof Fleszar, Stephen Kobourov, Joachim Spoerhase, Sankar
  Veeramoni, and Alexander Wolff.
\newblock Approximating the generalized minimum {M}anhattan network problem.
\newblock {\em Algorithmica}, 80(4):1170--1190, 2018.

\bibitem{das2015approximating}
Aparna Das, Emden~R Gansner, Michael Kaufmann, Stephen Kobourov, Joachim
  Spoerhase, and Alexander Wolff.
\newblock Approximating minimum {M}anhattan networks in higher dimensions.
\newblock {\em Algorithmica}, 71(1):36--52, 2015.

\bibitem{demaine2009geometry}
Erik~D Demaine, Dion Harmon, John Iacono, Daniel Kane, and Mihai
  P{\u{a}}tra{\c{s}}cu.
\newblock The geometry of binary search trees.
\newblock In {\em Proceedings of the twentieth annual ACM-SIAM symposium on
  Discrete algorithms}, pages 496--505. SIAM, 2009.

\bibitem{demaine2007dynamic}
Erik~D Demaine, Dion Harmon, John Iacono, and Mihai P{\u{a}}tra{\c{s}}cu.
\newblock Dynamic optimality—almost.
\newblock {\em SIAM Journal on Computing}, 37(1):240--251, 2007.

\bibitem{dodis1999design}
Yevgeniy Dodis and Sanjeev Khanna.
\newblock Design networks with bounded pairwise distance.
\newblock In {\em Proceedings of the thirty-first annual ACM symposium on
  Theory of computing}, pages 750--759, 1999.

\bibitem{feldman2012improved}
Moran Feldman, Guy Kortsarz, and Zeev Nutov.
\newblock Improved approximation algorithms for directed {S}teiner forest.
\newblock {\em Journal of Computer and System Sciences}, 78(1):279--292, 2012.

\bibitem{Golumbic2004perfect}
Martin~Charles Golumbic.
\newblock {\em Algorithmic Graph Theory and Perfect Graphs (Annals of Discrete
  Mathematics, Vol 57)}.
\newblock North-Holland Publishing Co., NLD, 2004.

\bibitem{Gudmundsson}
Joachim Gudmundsson, Christos Levcopoulos, and Giri Narasimhan.
\newblock Approximating a minimum {M}anhattan network.
\newblock {\em Nordic J. of Computing}, 8(2):219--232, June 2001.
\newblock URL: \url{http://dl.acm.org/citation.cfm?id=766533.766536}.

\bibitem{guo2008fast}
Zeyu Guo, He~Sun, and Hong Zhu.
\newblock A fast 2-approximation algorithm for the minimum {M}anhattan network
  problem.
\newblock In {\em International Conference on Algorithmic Applications in
  Management}, pages 212--223. Springer, 2008.

\bibitem{harmon2006new}
Dion Dion~Kane Harmon.
\newblock {\em New bounds on optimal binary search trees}.
\newblock PhD thesis, Massachusetts Institute of Technology, 2006.

\bibitem{iacono2016weighted}
John Iacono and Stefan Langerman.
\newblock Weighted dynamic finger in binary search trees.
\newblock In {\em Proceedings of the twenty-seventh annual ACM-SIAM symposium
  on Discrete algorithms}, pages 672--691. SIAM, 2016.

\bibitem{sleator1985self}
Daniel~Dominic Sleator and Robert~Endre Tarjan.
\newblock Self-adjusting binary search trees.
\newblock {\em Journal of the ACM (JACM)}, 32(3):652--686, 1985.

\bibitem{Snoeyink07}
Jack Snoeyink.
\newblock Maximum independent set for intervals by divide and conquer with
  pruning.
\newblock {\em Networks}, 49(2):158--159, 2007.
\newblock URL: \url{https://onlinelibrary.wiley.com/doi/abs/10.1002/net.20150},
  \href
  {http://arxiv.org/abs/https://onlinelibrary.wiley.com/doi/pdf/10.1002/net.20150}
  {\path{arXiv:https://onlinelibrary.wiley.com/doi/pdf/10.1002/net.20150}},
  \href {http://dx.doi.org/10.1002/net.20150} {\path{doi:10.1002/net.20150}}.

\bibitem{wilber1989lower}
Robert Wilber.
\newblock Lower bounds for accessing binary search trees with rotations.
\newblock {\em SIAM journal on Computing}, 18(1):56--67, 1989.

\end{thebibliography}
\appendix

\section{BST and Manhattan Problems}%
\label{sec: BST and Man}

\paragraph{MinASS problem} Let $p,q \in {\mathbb R}^2$ be points that are not horizontally or vertically aligned. Denote by $\Box_{p,q}$ the closed rectangle with two of its corners being $p$ and $q$.
We say that this rectangle is $P$-empty if $P \cap \Box_{p,q} = \{p,q\}$; otherwise, we say that it is $P$-satisfied.
We say that a collection of points $P \subseteq {\mathbb R}^2$ is \textit{arboreally satisfied} if, for any pair $p,q \in P$, the rectangle $\Box_{p,q}$ is $P$-satisfied.
In the MinASS problem, we are given a collection of points $P \subseteq {\mathbb R}^2$ such that no two points are horizontally or vertically aligned, and our goal is to compute a set $Q \supseteq P$ such that $Q$ is arboreally satisfied.

\begin{theorem}
Any set $Q \subseteq {\mathbb R}^2$ is arboreally satisfied if and only if there is a Manhattan path connecting every pair $p,q \in Q$.
\end{theorem}
\begin{proof}
The ``if'' direction is quite obvious: Suppose there is a Manhattan path connecting every pair. Then, consider any $p,q \in Q$ that is not aligned and $\Box_{p,q}$. Consider any point $v \not \in \{p,q\}$ on the Manhattan path $Z$ connecting $p$ to $q$. Clearly, $v \in \Box_{p,q}$.

For the ``only if'' direction, assume that the set $Q$ is arboreally satisfied but (for contradiction) not Manhattan-connected.
Consider points $p,q \in Q$ that are not $M$-connected (if there are many such pairs, choose one that minimizes the $\ell_1$-distance $||p-q||_1$.)
$\Box_{p,q}$ contains some point $q' \not\in \{p,q\}$ and point $q'$ is not at a corner of $\Box_{p,q}$ (otherwise, there would be a Manhattan path $p \rightarrow q' \rightarrow q$.) Therefore, $p$ is not aligned with $q'$ or $q$ is not aligned with $q'$. Assume it is the former (the other case is similar). Then $p$ and $q'$ is not connected, thus contradicting the choice of $(p,q)$ since~$||p-q'||_1 < ||p-q||_1$.
\end{proof}

\section{Proof of Lemma~\ref{lem:opt_vs_ir} }
\label{sec:proof of new lower bound}
	We repeat the proof from \cite{demaine2009geometry} and first show that $\ir(P,D) \leq \vs(P,D)$.
	Consider a set of independent monotone demands $R_j=R(p_j,q_j)$, for~$j=1,\ldots,k$, ordered decreasingly by width.
	Consider the rectangle $R_j$, $j>1$, which, among all rectangles intersecting the interior of $R_1$, has minimal~$x(p_j)$. Without loss of generality, we may assume that such a rectangle exists (otherwise $R_1$ can be separated trivially), and if there are several such rectangles, choose $j$ as small as possible.

	We use the observation that if two non-conflicting rectangles intersect, then one of them crosses the other on two of its sides, either top and bottom, or left and right.
	If $p_1 \neq p_j$, then $x(p_j) > x(p_1)$. Otherwise, $R_j$ must cross the left and right side of $R_1$, a contradiction to $R_1$ being the widest rectangle.
	Hence, we can place a line $\ell_1$ connecting top and bottom boundaries of $R_1$ just right of~$p_1$.
	If $p_1 = p_j$ on the other hand, then, by monotonicity and the previous observation, $x(q_j)<x(q_1)$ and~$y(q_j)>y(q_1)$.
	We show that there cannot be another rectangle whose interior intersects both~$R_1$ and the right border of $R_j$.
	Suppose there was such a rectangle $R_k$.
	If $R_k$ crosses the left and right sides of $R_j$ this contradicts the minimality of $x(p_j)$ for $R_j$.
	Therefore, the interior of $R_k$ could not have intersected the right border of $R_j$.
	Hence, we may place $\ell_1$ just right of $q_j$.
	Repeating the argument for the remaining rectangles, we obtain the statement.

	To show the second inequality, $\vs(P,D) \leq \opt(P,D)$, consider a vertically separable set of monotone demands, the corresponding demand rectangles $R_1,R_2,\ldots,R_k$ and their respective separating lines $\ell_1,\ell_2,\ldots,\ell_k$.
	Consider an optimal solution $Q$ that satisfies these demands as well as the Manhattan path it uses to satisfy $R_i$, that is the Manhattan-path connecting $p_i$ and $q_i$.
	This path must cross $\ell_i$ at some point horizontally. We associate with $R_i$ the pair $a_i,b_i \in P \cup Q$ of points directly right and left of $\ell_1$ on that path respectively.
	By the definition of separable demands, for two such point pairs $(a_i,b_i)$, $(a_j,b_j)$ that lie on the same horizontal line, either $x(a_i) < x(b_i) \leq x(a_j) < x(b_j)$, or $x(a_j) < x(b_j) \leq x(a_i) < x(b_i)$.
	This requires the monotonicity of the demands since it is otherwise possible that the top of some rectangle touches the bottom of another and the same point-pair is considered twice.
	We conclude that $k$ point-pairs on the same horizontal line consist of at least~$k+1$ distinct points.
	Moreover, since input-points lie on distinct rows, at most one of these is an input point.
	We can thus charge each demand $R_i$ to a distinct auxiliary point in the optimal solution.
	In fact, we can choose this to be either $a_i$ or $b_i$.

\section{Tightness of Theorem~\ref{thm:xcoord_approx}}%
\label{sec:tightness of is}

Let $s \in {\mathbb N}$ be any integer.
We show that there exists an $n$-point instance $(P,D)$ on $s$ columns such that $\is(P,D) = \Theta(n)$ while $\opt(P,D) = \Omega(n \log s)$.
We remark that this tightness holds even when $D$ is uniform.

There are many ways to prove this result, for instance, we can use the analysis of Wilber~\cite{wilber1989lower} adapted to the case with only $s$ columns. We sketch the proof here.

\newcommand{\wilber}{{\mathcal W}}

Let $\wilber(P)$ denote Wilber's first bound on $P$. It was proved in~\cite{demaine2009geometry} that $\wilber(P) \leq O(\opt(P,D))$, where $D = P \times P$ represents complete demands.

\begin{theorem}
Let $P$ be a random sequence of size $n$ on $s$ columns. Then $\wilber(P) =\Omega(n \log s)$ with high probability.
\end{theorem}

In particular, there exists an input $P$ with $\wilber(P) = \Omega(n \log s)$.
Since $\is(P,D) \leq n$, we have our desired result.

\section{NP-Hardness Proof for \problemname}%
\label{sec:app:np-hardness}
\label{sec:hardness_proof}

We first observe that $\alpha$ – the number of non-$SC$ demands $D_{\overline{SC}}$ – has polynomial size, namely
\begin{equation}
  \alpha
= 3m \cdot 4 + n \cdot 4
= 12m + 4n
.
\end{equation}
To see this, note that each of the $m$ clause gadgets $GC_j$ contains $3$ literal points $\ell_{jk}$.
Such an~$\ell_{jk}$ \enquote{connects} its corresponding positive or negative variable literal point ($x_i^+$ or $x_i^-$) to the clause point~$c_j$ via a dedicated chain of $4$ demands (going through a dedicated literal connector and a dedicated connection point).
Moreover, for each of the $n$ variable gadgets $GX_i$, we have a demand~$(S, x_i)$, a demand $(x_i, d_i)$, and two demands $(x_i^+, d_i)$ and $(x_i^-, d_i)$.
Together, these account for all non-$SC$ demands.

We start with the easier direction of the reduction, showing that a satisfiable assignment of $\phi$ gives rise to a \MGMC/ solution of size $\alpha$.
\begin{lemma}%
\label{lem:phisat}
If $\phi$ is satisfiable, then there exist solutions to $(P_{\phi}, D_{\phi})$ of size $\alpha$.
\end{lemma}
\begin{proof}
Consider the family of (partial) solutions $Q_{\phi}$ constructed as follows:
For each of the $\alpha - n$ non-$SC$ and non-$XD$ demands $d \in D_{\overline{SC}} \setminus \set{(x_i, d_i) | i \in \intcc{n}}$, add the top-left corner of demand $d$ to $Q_{\phi}$, ensuring that $d$ is trivially satisfied.
For the $n$ $XD$ demands $(x_i, d_i)$ however, we add \emph{either} the top-left corner of the demand $(x_i, x_i^+)$ \emph{or} of the demand $(x_i, x_i^-)$ to $Q_{\phi}$.
Note that in both cases, since $x_i^+$ and $x_i^-$ each have already established a trivial connection to $d_i$, we get a \mpath/ from $x_i$ to $d_i$, satisfying the $XD$ demand $(x_i, d_i)$.

Any (partial) solution $Q_{\phi}$ constructed in this way has size $\alpha$ (it contains one new point for each non-$SC$ demand), satisfies all non-$SC$ demands, and may or may not satisfy some $SC$ demands.
Moreover, there is a natural one-to-one mapping between such partial solutions and variable assignments for $\phi$ (set $X_i = 1$ if and only if $(x_i, d_i)$ is satisfied via $x_i^+$).
Now consider which $SC$ demands are satisfied by such a solution/variable assignment $Q_{\phi}$.
Note that $Q_{\phi}$ yields a \mpath/ from $x_i^+$ to~$c_j$ \emph{if and only if} $C_j$ contains the positive literal $X_i$ (the \mpath/ can then use the corresponding positive variable-clause path).
Similarly, $Q_{\phi}$ produces a \mpath/ from $x_i^-$ to $c_j$ \emph{if and only if} $C_j$ contains the negative literal $\neg X_i$.
On the other hand, $Q_{\phi}$ yields a \mpath/ from $S$ to $x_i^+$ \emph{if and only if} the variable assignment $Q_{\phi}$ sets $X_i = 1$ and a \mpath/ from $S$ to $x_i^-$ otherwise.
Together, we get a \mpath/ from $S$ to $c_j$ if and only if $C_j$ is satisfied via $X_i$ in the variable assignment $Q_{\phi}$.
As a consequence, the solution $Q_{\phi}$ satisfies exactly those $SC$ demands $(S, c_j)$ for which $C_j$ is satisfied by the variable assignment $Q_{\phi}$.
In summary, $Q_{\phi}$ is a boolean solution (as defined in \cref{sec:np_hardness}).

Now, if $\phi$ is satisfiable, fix a solution $Q_{\phi}$ that corresponds to a satisfying variable assignment.
From the above it follows that $Q_{\phi}$ has size $\alpha$ and satisfies all demands in $D_{\phi}$.
\end{proof}

The remainder of this section proves the remaining direction, namely that any solution of size~$\alpha$ implies that $\phi$ is satisfiable (\cref{lem:phinotsat}).
As this turns out to be more involved, we require some preparation.
First, we provide a simple result about the structure of \enquote{triangular} instances.
Afterwards, we describe a partitioning of $(P_{\phi}, D_{\phi})$ and $Q_{\phi}$ into suitable subinstances/-solutions (some of which are triangular).
This allows us to argue separately about the structure of these subinstances (triangular instances and the variable-clause paths), which will help us to get the desired result.

\paragraph{Triangular Instances}
A \emph{triangular instance} $(P_{\Delta}, D_{\Delta})$ of size $n \in \N$ consists of $n + 1$ input points $P_{\Delta} = \set{p_0, p_1, \dots, p_n}$ and of $n$ demands $D_{\Delta} = \set{(p_0, p_i) | i \in \intcc{n}}$.
The points~$p_1, p_2, \dots, p_n$ form a descending diagonal and lie to the top-right of $p_0$.
We refer to the point set~$G_{\Delta} \coloneqq \set{(x, y) | \exists i, j \in \intcc{n}_0\colon x = x(p_i), y = y(p_j)} \cap R(D_{\Delta})$ as the \emph{(triangular) grid} of $(P_{\Delta}, D_{\Delta})$.

Optimal solutions for triangular instances have size $n$, since there are solutions of that size (e.g., the $n$ grid points above $p_0$) and the $n$ demands from $D_{\Delta}$ are monotone and vertically separable (implying a lower bound of $n$ by \cref{lem:opt_vs_ir}).
The next \lcnamecref{prop:triangular_demands} states that any optimal solution to a triangular instance lies on the instance's triangular grid.
\begin{proposition}%
\label{prop:triangular_demands}
Consider a triangular instance $(P_{\Delta}, D_{\Delta})$ of size $n \in \N$.
If $Q_{\Delta}$ has size $n$ and satisfies $D_{\Delta}$, then $Q_{\Delta}$ is a subset of the instance's grid $G_{\Delta}$.
\end{proposition}
\begin{proof}
The \lcnamecref{prop:triangular_demands} holds trivially for $n = 1$, since then the sole solution point must lie either in the top-left or in the bottom-right corner of the sole demand $(p_0, p_1)$.
For the case of a contradiction, assume the \lcnamecref{prop:triangular_demands} is not true for all $n \in \N \setminus \set{1}$ and fix the smallest~$n$ for which this is the case.
So there is a triangular instance $(P_{\Delta}, D_{\Delta})$ of size $n$ as well as a corresponding solution $Q_{\Delta}$ of size $n$ that contains a non-grid point $q \in Q_{\Delta} \setminus G_{\Delta}$.
Without loss of generality, assume $p_0 = (0, 0)$.

Depending on the position of $q$, transform $(P_{\Delta}, D_{\Delta})$ into a suitable triangular instance~$(P'_{\Delta}, D'_{\Delta})$ of size $n-1$ with a corresponding solution $Q'_{\Delta}$ of size at most $n-1$ containing a non-grid point.
Once this is achieved, we immediately get a contradiction to the minimality of $n$.

Consider first the case $x(q) \geq x(p_1)$.
We transform $(P_{\Delta}, D_{\Delta})$ into $(P'_{\Delta}, D'_{\Delta})$ by removing both the input point $p_1$ and the demand $(p_0, p_1)$ and by projecting $p_0$ to the $x$-coordinate $x(p_1)$ (not changing its $y$-coordinate).
To construct $Q'_{\Delta}$ from $Q_{\Delta}$ we project all points with $x$-coordinate at most~$x(p_1)$ to $x$-coordinate $x(p_1)$ and remove any points with $y$-coordinate $> y(p_2)$ (which cannot help to connect $p_0$ to any of $p_2, p_3, \dots, p_n$).
One can easily check that $Q'_{\Delta}$ is a solution for $(P'_{\Delta}, D'_{\Delta})$:
Any \mpath/ going through one of the projected points remains intact, as the projected point is still reachable from the projection of $p_0$ (by going straight up).
Moreover, $\abs{Q'_{\Delta}} < \abs{Q_{\Delta}}$, since we either removed the top-left grid point if it was in~$Q_{\Delta}$ (since it has $y$-coordinate $y(p_1) > y(p_2)$) or at least two of the projected points had the same $y$-coordinate (or there could not be a \mpath/ from $p_0$ to~$p_1$ in $Q_{\Delta}$), causing them to get merged during the projection.
Note that $q \in Q'_{\Delta}$, since this case assumes $x(q) \geq x(p_1)$, such that $q$ is not affected by the projection.
As stated above, we get a contradiction to the minimality of $n$.

The case $y(q) \geq y(p_n)$ yields the same contradiction via a symmetrical argument.
Thus, it remains to consider the case that $q$ and all other non-grid points have $x$-coordinate
$< x(p_1)$ and~$y$-coordinate $< y(p_n)$.
Then there must we some \mpath/ that leaves either the $x$- or the $y$-axis at some non-grid point $q_1 \in Q_{\Delta}$ to reach another non-grid point $q_2 \in Q_{\Delta}$.
We consider only the case that $q_1$ lies on the $y$-axis; the other case is proven symmetrically.
We transform $(P_{\Delta}, D_{\Delta})$ into~$(P'_{\Delta}, D'_{\Delta})$ using the same construction as above (projecting everything left of $p_1$ onto $x(p_1)$).
This causes~$q_1$ and $q_2$ to merge (giving $\abs{Q'_{\Delta}} < \abs{Q_{\Delta}}$) and ensures the existence of a non-grid point in $Q'_{\Delta}$ (the projection of $q_1$ and $q_2$).
Again, as stated above, this yields a contradiction to the minimality of~$n$.
\end{proof}

\paragraph{Partitioning the Instance \& Solution}
Consider a solution $Q_{\phi}$ to $(P_{\phi}, D_{\phi})$ with $\abs{Q_{\phi}} = \alpha$.
We partition $(P_{\phi}, D_{\phi})$ and $Q_{\phi}$ into suitable subinstances $(P_{\bullet}, D_{\bullet})$ and subsolutions $Q_{\bullet}$ such that the different $(P_{\bullet}, D_{\bullet})$ cover different areas, $Q_{\bullet}$ solves $(P_{\bullet}, D_{\bullet})$, and $\abs{Q_{\bullet}} = \abs{D_{\bullet}}$.
This will allow us to argue about the structure of these different subsolutions in the proof of \cref{lem:phinotsat} (e.g., by realizing that one of the subinstances is a triangular solution, such that we can apply \cref{prop:triangular_demands}).
For our partitioning, define
\begin{itemize}[noitemsep]
\item $D_{SC} = \set{(S, c_j) | j \in \intcc{m}}$ (the $m$ $SC$ demands),
\item $D_{SX} = \set{(S, x_i) | i \in \intcc{n}}$ (the $n$ $SX$ demands),
\item $D_{cl} = \set{(\ell_{jk}, c_j) | j \in \intcc{m}, k \in \intcc{3}}$ (the $3m$ demands in clause gadgets),
\item $D_{con} = \set{d \in D_{\phi} | \exists i, j\colon p_{ij}^+ \in d \lor p_{ij}^- \in d}$ (the $6m$ connection demands), and
\item $D_{var} = \set{d \in D_{\phi} | \exists i \in \intcc{n}\colon \set{x_i^+, x_i^-, d_i} \cap d \neq \emptyset}$ (the $3m + 3n$ demands in variable gadgets).
\end{itemize}
With this, we have $D_{\overline{SC}} = D \setminus D_{SC} = D_{SX} \cupdot D_{cl} \cupdot D_{con} \cupdot D_{var}$.
Remember that $\alpha = \abs{D_{\overline{SC}}}$ (by definition of $\alpha$) and note that $D_{\overline{SC}}$ is monotone and vertically separable (all non-$XD$ demands intersect only at corners, and the $XD$ demands can be separated by a vertical line segment immediately to the right of the variable point $x_i$).
Moreover, $Q_{\phi}$ also has size $\alpha$ and satisfies the demands~$D_{\overline{SC}} \subseteq D_{\phi}$.
Thus, \cref{cor:coveringsubsolutions} implies $Q_{\phi} \subseteq R(D_{\overline{SC}})$.
The pairwise intersections of the four areas $R(D_{\bullet})$ ($\bullet \in \set{SX, cl, var, con}$) contain only input points.
Thus, we get a natural partition $Q_{\phi} = Q_{SX} \cupdot Q_{cl} \cupdot Q_{var} \cupdot Q_{con}$, with $Q_{\bullet} = \set{q \in Q_{\phi} | q \in R(D_{\bullet})}$, such that $Q_{\bullet}$ satisfies~$D_{\bullet}$.
With this, we can apply \cref{lem:opt_vs_ir} to get $\abs{Q_{\bullet}} = \abs{D_{\bullet}}$ (since $\abs{Q_{\bullet}} \geq \abs{D_{\bullet}}$ by \cref{lem:opt_vs_ir}, and, when summing over $\bullet \in \set{SX, cl, var, con}$, each side of this inequality sums up to $\alpha$).
To finalize the partition, let~$P_{\bullet} \coloneqq \set{p \in P_{\phi} | p \in R(D_{\bullet})}$ (which strictly is not a partition of $P_{\phi}$, since, e.g.,~$x_i \in P_{SX} \cap P_{var}$).

\paragraph{Routing of Variable-Clause Paths}
We now use the above partitioning to argue about the structure of $Q_{con}$, showing that $Q_{con}$ consists of exactly one corner from each $d \in D_{con}$.
We will use this to show that $Q_{\phi}$ cannot \enquote{cheat} by jumping between different variable-clause paths, but instead must enter a variable-clause path at its entrance (at some variable gadget) and leave it at its exit (at some clause gadget).
\begin{proposition}%
\label{prop:connection_demands}
$Q_{con}$ satisfies each demand $d \in D_{con}$ via a dedicated point $q_d \in Q_{con}$ that lies in the top-left or bottom-right corner of $d$.
\end{proposition}
\begin{proof}
The connection demands in $D_{con}$ correspond to the horizontal and vertical paths connecting variable and clause gadgets to their connection points, see \cref{fig:construction:overview}.
We call a connection demand~$d \in D_{con}$ either \emph{horizontal} or \emph{vertical}, depending on whether $d$ leaves a variable/clause gadget horizontally or vertically.
Our construction ensures that no point is covered by more than two connection demands and that both the top-left and bottom-right corner of a connection demand are not contained in any other connection demand.
Moreover, if two connection demands intersect, one is horizontal and one is vertical.

Assume the statement is not true, so there is a $d \in D_{con}$ for which neither its top-left nor its bottom-right corner is in $Q_{con}$.
We consider only the case that $d$ is a horizontal connection demand; the vertical case is symmetrical.
Then there must be two points $q_{1}, q_{2} \in Q_{con} \cap R(d)$ that share the same $x$-coordinate.
Both of them must be covered by a vertical connection demand $d'$, since otherwise we found two points in $Q_{con}$ covered by only one demand from $D_{con}$, a contradiction to \cref{cor:coveringsubsolutions}.
If $Q_{con}$ contains the top-left or bottom-right corner of $R(d')$ (which cannot be covered by any other connection demand), we found three points in $Q_{con}$ ($q_{1}$, $q_{2}$, and the corner) covered by only two demands from $D_{con}$ ($d$ and $d'$), which is again a contradiction to \cref{cor:coveringsubsolutions}.
Thus,~$Q_{con}$ cannot contain a corner of $R(d')$, which implies two points $q_{1}', q_{2}' \in Q_{con} \cap R(d')$ that share the same~$y$-coordinate.
If $\set{q_{1}, q_{2}} \cap \set{q_{1}', q_{2}'} \neq \emptyset$, this intersection has size exactly three.
But then we found three points in $R(d) \cap R(d')$, an area that cannot be intersected by a third connection demand, yielding again a contradiction to \cref{cor:coveringsubsolutions}.
So we must have $\set{q_{1}, q_{2}} \cap \set{q_{1}', q_{2}'} = \emptyset$.
This yields the final contradiction to \cref{cor:coveringsubsolutions}, since we found four points in $Q_{con}$ covered by at most three connection demands from $D_{con}$ ($d$, $d'$, and a potential third horizontal connection demand covering~$q_{1}'$ and $q_{2}'$).
\end{proof}

With this, we are ready to prove the second direction of the reduction:
\begin{lemma}%
\label{lem:phinotsat}
If there is a solution to $(P_{\phi}, D_{\phi})$ of size $\alpha$, then $\phi$ is satisfiable.
\end{lemma}
\begin{proof}
Consider a solution $Q_{\phi}$ to $(P_{\phi}, D_{\phi})$ with $\abs{Q_{\phi}} = \alpha$.
We basically show that $Q_{\phi}$ is a boolean solution, such that it corresponds to a satisfying variable assignment for $\phi$ (since it satisfies all $SC$ demands).
To this end, remember the partitioning of $(P_{\phi}, D_{\phi})$ and $Q_{\phi}$ described above.
As detailed above, each $(P_{\bullet}, D_{\bullet})$ with solution $Q_{\bullet}$ satisfies the preconditions of \cref{cor:coveringsubsolutions}.
We use this to argue about the structure of the different $Q_{\bullet}$.

We start with the structure of $Q_{SX}$.
Observe that $(P_{SX}, D_{SX})$ is a triangular instance of size~$n$ that is satisfied by the solution $Q_{SX}$ of size $n$.
By \cref{prop:triangular_demands}, $Q_{SX}$ lies on the triangular grid of~$(P_{SX}, D_{SX})$.
As an immediate consequence, any \mpath/ that satisfies an $SC$ demand $(S, c_j)$ must go through one of the variable points $x_i$ (note that no $q \in Q_{\phi} \setminus Q_{SX}$ can share an~$x$- or~$y$-coordinate with $S$, since none of the demands~$d \in D_{\overline{SC}} \setminus D_{SX}$ would cover such a $q$).

Next, remember that \cref{prop:connection_demands} restricts $Q_{con}$ to the corners of $D_{con}$.
This implies that a \mpath/ entering $GX_i$ through $x_i$ cannot immediately leave~$GX_i$ but must first go through the inside of $GX_i$.
Otherwise, there would be some $q \in Q_{\phi}$ outside of $GX_i$ either with $x(q) = x(x_i)$ and $y(q) > y(x_i)$ or with $y(q) = y(x_i)$ and $x(q) > x(x_i)$.
This $q$ must be covered by some demand~$d \in D_{\overline{SC}}$, and our construction is such that this can only be done for a $d \in D_{con}$ (see \cref{fig:construction:overview}).
That is $q \in Q_{con}$, and \cref{prop:connection_demands} then implies that $q$ lies in the corner of $d$.
But no corner of a connection demand shares an $x$- or $y$-coordinate with $x_i$ (see \cref{fig:construction:overview}), such that $q$ cannot exist.

A simple consequence of this is that any \mpath/ entering variable~$GX_i$ (through $x_i$) can leave~$GX_i$ only by entering one of the connection demands $(x_{ik}^+, p_{ij}^+)$ or $(x_{ik}^-, p_{ij}^-)$ (see \cref{fig:gadgets:variable}) and must then follow the corresponding variable-clause path to $GC_j$ (see \cref{fig:construction:overview}).

So in summary, any \mpath/ from $S$ to a clause point $c_j$ must enter some variable gadget $GX_i$ at $x_i$, which it then must leave either through a positive or negative variable-clause path to reach a connected clause gadget $GC_j$.
Note that $Q_{cl}$ guarantees that this \mpath/ actually reaches $c_j$ (in fact, \cref{prop:triangular_demands} implies a strict structure on $Q_{cl}$, but we do not need this for our proof).
In the next and final step, we argue about the structure of $Q_{var}$, showing that $x_i$ cannot be connected to both a positive and negative variable-clause path.
Once this is shown, we can identify $Q_{\phi}$ with the variable assignment that assigns $X_i = 1$ if and only if $x_i$ is connected to positive variable-clause paths and $X_i = 0$ otherwise.
This implies that $Q_{\phi}$ is a boolean solution and, since $Q_{\phi}$ satisfies all~$SC$ demands, that $\phi$ is satisfiable.

So consider $Q_{var}$.
For any $i \in \intcc{n}$ there must be a dedicated $q_i \in Q_{var}$ on the left or bottom side of the $XD$ demand $(x_i, d_i) \in D_{var}$.
We consider only the case that $q_i$ lies on the bottom side; the other case is symmetrical.
Since $(x_i, d_i)$ is the only non-$SC$ demand that covers $q_i$, \cref{cor:coveringsubsolutions} implies that the remaining $\abs{Q_{var}} - n$ solution points from $Q_{var}$ lie in $R(D_{var} \setminus D_{XD})$, where $D_{XD} \coloneqq \set{(x_i, d_i) | i \in \intcc{n}}$.
Thus, the inside of the rectangle formed by~$x_i$ and the point~$(x(x_i^-), y(x_i^+))$ is empty (see \cref{fig:gadgets:variable}).
But then, the next point $q'_i$ on the \mpath/ coming from~$x_i$ through $q_i$ must have $y(q'_i) \geq y(x_i^+)$ or $x(q'_i) \geq x(x_i^-)$.
In the former case, $q'_i$ (which is used by any \mpath/ coming from $x_i$) can reach only positive variable-clause paths, while in the latter case only negative variable clause-paths are reachable.
As argued above, this finishes the \lcnamecref{lem:phinotsat}'s statement.
\end{proof}

\section{Sublogarithmic Approximation for \problemname}
\label{sec:sublogarithmic_general_approx:analysis}

In this section, we present the algorithm underlying Theorem~\ref{thm:general_approx}, that is, an algorithm with approximation ratio $O(\sqrt{\log n})$.
The idea is to divide the plane into $s\ll n$ many vertical strips, where $s$ can be thought of as a function that is sublinear in n. 
We call demands between points lying within the same strip \emph{intra-strip demands} and between points in different strips \emph{inter-strip demands}.
In order to satisfy the intra-strip demands, we recursively execute the algorithm (with fixed recursion parameter~$s$) and exploit that demand pairs in different strips are independent. To satisfy the inter-strip demands, we project points to the adjacent strip boundaries. This results in an instance with~$s\ll n$ many distinct $x$ coordinates for which we can apply the algorithm {\sc HorizontalManhattan} presented in the previous section.

\subsection{Algorithm Description}%

Before giving an overview of Algorithm~\ref{alg:generic-algorithm}, we discuss its building blocks and some notation.

\paragraph{Input}
The input consists of an instance $(P,D)$ of \MGMC/ and a parameter $s$ specifying the number of sub-problems (strips) in which the algorithm should divide the instance. As in the previous section, we allow that distinct points in $P$ have the same $x$-coordinate but require distinct~$y$-coordinates. Again, 
we may assume monotone demands. 

\paragraph{Subdivision into Strips}
Let $(P,D)$ be an instance of \MGMC/. Let $s\geq 2$ be a positive integer and let $B$ be a set of $s-1$ distinct vertical lines in the plane. The lines in~$B$ partition the plane into a set $\mathcal{S}$ of $s$ many vertical strips. We call the lines in $B$ \emph{strip boundaries} of $\mathcal{S}$, where each strip has one or two strip boundaries (left and right). We call $\mathcal{S}$ a \emph{strip subdivision of $(P,D)$} if each boundary of $\mathcal{S}$ is disjoint from $P$.

Let $g$ be the number of $x$-groups in $P$. In line~\ref{alg:divide-strips}, we compute a strip division $\mathcal{S}$ such that each strip contains at most $\lceil g/s\rceil$ many $x$-groups. We call such a strip subdivision \emph{balanced}.

\paragraph{Intra-Strip Instances}
We first satisfy intra-strip demand pairs (Algorithm~\ref{alg:generic-algorithm}, Line~\ref{alg:intra-strip-recursion}).
To this end, we invoke the algorithm recursively for each strip $S$ and solve the sub-instance of~$(P,D)$, which satisfies all demands in $D$ lying completely in $S$ (intra-strip instance). The result is added to our solution $Q$.

\paragraph{Inter-Strip Instances}
Consider an instance $(P,D)$ of \MGMC/ and a strip subdivision~$\mathcal{S}$ with strips $S_1,\dots,S_s$ ordered left-to-right. For each input point $p\in P$, there is a unique strip $S\in\mathcal{S}$ containing $p$. If $S\neq S_1$, then we denote by $\pi^{\textnormal{l}}_{\mathcal{S}}(p)$ the horizontal projection of $p$ to the left boundary of $S$. Analogously, if $S\neq S_s$, then we denote by $\pi^{\textnormal{r}}_{\mathcal{S}}(p)$ the horizontal projection of $p$ to the right boundary of $S$. For any subset $P'\subseteq P$, we define $\pi_{\mathcal{S}}^{\textnormal{l}}(P')=\{\,\pi^{\textnormal{l}}_{\mathcal{S}}(p)\mid p\in P'\setminus S_1\,\}$, $\pi_{\mathcal{S}}^{\textnormal{r}}(P')=\{\,\pi^{\textnormal{r}}_{\mathcal{S}}(p)\mid p\in P'\setminus S_s\,\}$, and $\pi_{\mathcal{S}}(P')=\pi_{\mathcal{S}}^{\textnormal{l}}(P')\cup \pi_{\mathcal{S}}^{\textnormal{r}}(P')$ as the set of horizontal projections of the points in $P'$ to left, right, and both strip boundaries of their respective strips in $\mathcal{S}$, respectively.

Define a new instance on the projection $\pi_{\mathcal{S}}(P)$ of $P$ to the adjacent strip boundaries by
\begin{displaymath}
\pi_{\mathcal{S}}(D)=\{\,(\pi_{\mathcal{S}}^{\textnormal{r}}(p),\pi_{\mathcal{S}}^{\textnormal{l}}(q))\mid (p,q)\in D\textnormal{ and $p,q$ lie in different and non-adjacent strips}\},
\end{displaymath}
the ``projected'' demands. (Recall that we assumed that $x(p)<x(q)$.) We drop demands within the same strip as they are handled in the respective intra-strip instance and demands from adjacent strips because are projected to the same boundary and thus automatically satisfied. Call $\pi_{\mathcal{S}}(P,D):=(\pi_{\mathcal{S}}(P),\pi_{\mathcal{S}}(D))$ the \emph{inter-strip instance} corresponding to~$(P,D)$. Note that the input points of this resulting instance share $s-1$ many distinct $x$-coordinates.

Intuitively, the following lemma states that a feasible solution to the inter-strip distance along with the projected points satisfies all the demand-pairs for which not both points are in the same strip. This fact is used in algorithm and analysis.

\begin{lemma}\label{lem:boundary-instances}
  Let $(P,D)$ be an instance with strip subdivision $\mathcal{S}$. If $Q$ is a feasible solution to the inter-strip instance $\pi_{\mathcal{S}}(P,D)$, then $Q\cup\pi_{\mathcal{S}}(P)$ satisfies all demands $(p,q)\in D$ for which $p$ and $q$ lie in different strips.
\end{lemma}
\begin{proof}
  To see the feasibility of $Q\cup\pi_{\mathcal{S}}(P)$ consider a demand $(p,q)\in D$ where $p$ and $q$ lie in different strips. Observe that we can connect $p$ to $\pi_{\mathcal{S}}^{\textnormal{r}}(p)$ directly, then $\pi_{\mathcal{S}}^{\textnormal{r}}(p)$ to $\pi_{\mathcal{S}}^{\textnormal{l}}(q)$ via $Q$, and finally connect $\pi_{\mathcal{S}}^{\textnormal{l}}(q)$ to $q$ to obtain a Manhattan path from $p$ to $q$. Recall that $\pi_{\mathcal{S}}^{\textnormal{r}}(p)$ and $\pi_{\mathcal{S}}^{\textnormal{l}}(q)$ are automatically connected if $p$ and $q$ lie in adjacent strips.
\end{proof}

Algorithm~\ref{alg:generic-algorithm} provides an overview of our algorithm. 

\begin{algorithm}[t]
\SetKwInOut{Input}{input}
\SetKwInOut{Output}{output}
\Input{ \MGMC/ instance $(P,D)$ and a positive integer $s\geq 2$}
\Output{ Feasible solution $Q\subseteq\mathbb{R}^2$ for $(P,D)$}
\lIf{there is only one $x$-group}{\Return $\emptyset$}
divide the plane $\mathbb{R}^2$ into a collection $\mathcal{S}$ of $s$ vertical strips each of which contains roughly the same number of $x$-groups\label{alg:divide-strips}\;
\ForEach{\textnormal{strip} $S\in\mathcal{S}$}{
   $Q\gets Q\cup\textnormal{VerticalManhattan}(P\cap S,D\cap (S\times S),s)$\label{alg:intra-strip-recursion}\;
}
$(P',D')\gets \pi_{\mathcal{S}}(P,D)$\label{alg:inter-strip-create}\;
$Q\gets Q\cup \textnormal{HorizontalManhattan}(P',D')$\label{alg:inter-strip-connect}\;
\Return{$Q\cup\pi_{\mathcal{S}}(P)$\label{alg:projection}\;}
\caption{VerticalManhattan$(P,D,s)$\label{alg:generic-algorithm}}
\end{algorithm}

\subsection{Analysis} 
\label{sec:sublogarithmic_general_approx:algorithm}

We now proceed to analyze the algorithm. 
Let us repeat the subadditivity property of vertically separable demands shown earlier, stating that the number of vertically separable demands is at least as much as that of the sum of corresponding numbers in the inter-strip instance and all intra-strip instances together, and in turn at least as much as the sum of numbers of vertically separable demands in all inter-strip instances created by the algorithm.

\begin{corollary}\label{cor:strip-separability}
Let $\mathcal{M}$ be the collection of inter-strip instances $(P',D')$ created in line~\ref{alg:inter-strip-create} over all recursive calls of the algorithm \textnormal{RecursiveManhattan} applied to an instance $(P,D)$ of Generalized Manhattan Connections. We have that
\begin{displaymath}
  \vsep(P,D)\geq\sum_{(P',D')\in\mathcal{M}}\vsep(P',D')\,.
\end{displaymath}
\end{corollary}
\begin{proof}
  Apply Lemma~\ref{lem:strip-separability} recursively.
\end{proof}

\paragraph{Bounding the Cost}
We are now ready to prove that \cref{alg:generic-algorithm} yields the \lcnamecref{sec:sublogarithmic_general_approx}'s main result (\cref{thm:general_approx}).

\begin{proof}[Proof of \cref{thm:general_approx}]
    The feasibility of the solution computed by VerticalManhattan follows by inductively applying Lemma~\ref{lem:boundary-instances} in line \ref{alg:projection} and the fact that HorizontalManhattan actually computes a feasible solution to the inter-strip instance $(P',D')$ in line~\ref{alg:inter-strip-connect}.

    To bound the cost of the solution, consider the collection  $\mathcal{M}$ of inter-strip instances $(P',D')$ created in line~\ref{alg:inter-strip-create} over all recursive calls of the algorithm \textnormal{RecursiveManhattan} applied to input $(P,D,s)$.  By \cref{lem:inter-strip} and Corollary~\ref{cor:strip-separability}, we can bound the total cost of points added to our solution in line~\ref{alg:inter-strip-connect} for solving the inter-strip instances by
\begin{displaymath}
    \sum_{(P',D')\in\mathcal{M}}O(\log(s))\cdot\vsep(P',D')\leq O(\log(s))\cdot \vsep(P,D)\leq O(\log(s))\cdot\opt(P,D)\,.
  \end{displaymath}

To bound the cost of points added in~\cref{alg:projection} via projection, note that there are $O(\log n/\log s)$ levels of recursion. Since each point ends up in precisely one sub-instance upon the recursive calls in~\cref{alg:intra-strip-recursion} each original input point contributes at most two points to the projection computed in ~\cref{alg:projection} at each recursion level. Hence, the overall cost of this step is $O(\log n/\log s)\cdot |P|=O(\log n/\log s)\cdot\opt(P,D)$.

Adding these two cost components gives an overall cost of $O(\log s+\log n/\log s)\cdot\opt(P,D)$. The theorem follows by setting $s=2^{\sqrt{\log n}}$.
\end{proof}

\section{Towards Better Approximation}
\label{sec:complete-bipartite}
In this section, we state some structural properties and show that these would imply a better approximation ratio of $O(\log \log n)$.
Furthermore, we prove that those properties hold for a special class of demands.

\subsection{Structural Properties for $O(\log \log n)$-approximation}

Let $f$ be a lower bound function on $\opt$, that is $f(P,D) \leq O(\opt(P,D))$ for any input~$(P,D)$.
We say that $f$ satisfies {\em subadditivity} if for any 
strip partitioning $\sset$, we have
\[ f(\pi_{\sset}(P,D)) + \sum_{S \in \sset} f(P \cap S, D \cap (S \times S)) \leq f(P,D). \]
For any lower bound function $f$, an \textit{$f$-sparsification algorithm} is an efficient algorithm that, on input $(P,D)$, produces an instance $(\tilde{P},\tilde{D})$ such that (i) any feasible solution $Q$ for $(\tilde{P},\tilde{D})$ is feasible for $(P,D)$ as well, and (ii)  $|\tilde{P}|=O(f(\tilde{P},\tilde{D}))$.
The following theorem is given implicitly in~\cite{chalermsook2019pinning}. The proof is a verbatim adaptation of the original proof and therefore omitted.

\begin{theorem}
Let $f$ be a lower bound function. Assume that $f$ is subadditive and that there exists an $f$-sparsification algorithm. Then there exists an efficient algorithm for \problemname~producing a solution of cost at most $O(\log \log n) f(P,D) \leq O(\log \log n) \opt(P,D)$.
\label{thm: log log n}
\end{theorem}

We do not know how to prove the existence of sparsification algorithms for general instances and leave this as an interesting open problem.
However, 
we will show that such an algorithm exists when the demand graph is a complete $k$-partite 
graph.

The existence of a sparsification algorithm can be seen as a stronger property than the previously considered existence of an $O(\log s)$-approximation algorithm when input points have at most~$s$~distinct $x$-coordinates.

\begin{lemma}
If there exists an $f$-sparsification algorithm for an instance~$(P,D)$ with at most~$s$~distinct $x$-coordinates, then there exists an algorithm that returns a solution with cost at most $O(\log s) f(P,D)$.
\end{lemma}
\begin{proof}
The algorithm is as follows. First, apply the $f$-sparsification algorithm on $(P,D)$ to obtain~$(P',D')$. Specifically, we have that $|P'| \leq O(f(P,D))$.
We say that an $x$-coordinate is active if there is an input point on it.
Draw a vertical partitioning line $\ell$ in the middle such that there are~$\lceil s/2 \rceil$ active $x$-coordinates to the left of $\ell$ and the rest to the right.
Let~$\sset$ consist of the two strips~$S_L, S_R$ 
resulting from drawing the line $\ell$. That is the strips $S_L$ and $S_R$ are the regions on the left and right of line $\ell$ respectively.
For each point $p \in P'$, we add a point $p'$ to the solution that lies on $\ell$ and is horizontally aligned with $p$. Notice that this satisfies all demands crossing~$\ell$.
Next, we recurse inside $(P' \cap S_L, D' \cap (S_L \times S_L))$ and $(P' \cap S_R, D' \cap (S_R \times S_R))$. It is easy to see that this algorithm produces a feasible solution.

To analyze the cost, notice that there are $O(\log s)$ levels of recursions. The input instance is called \emph{level-$0$ instance} and creates $2$ sub-instances (called \emph{level-$1$ instances}) and so on until the last level instances.
There are at most $2^j$ level-$j$ instances. Let $\iset_j$ be the collection of all level-$j$ instances. Note that the instances in $\iset_j$ lie on $2^j$ different strips which are pairwise disjoint.
Additionally, each~$(\hat{P},\hat{D}) \in \iset_j$ creates $O(f(\hat{P},\hat{D}))$ points before making recursive calls, so the total cost of level-$j$ recursions is $O(\sum_{(\hat{P}, \hat{D}) \in \iset_j} f(\hat{P},\hat{D}))$. Since these instances $(\hat{P}, \hat{D}) \in \iset_j$ lie on disjoint strips, we have that $O(\sum_{(\hat{P}, \hat{D}) \in \iset_j} f(\hat{P},\hat{D})) \leq O(f(P,D))$.
\end{proof}

The algorithm we use to obtain an $O(\sqrt{\log n})$-approximation in Section~\ref{sec:sublogarithmic_general_approx} can be stated in its generality as follows.

\begin{theorem}
Let $f$ be a lower bound function. Assume that $f$ is subadditive and that there exists an algorithm for \problemname~that costs at most $O(\log s) f(P,D)$ for any $s$-thin instance $(P,D)$.
Then, there exists an efficient $O(\sqrt{\log n})$ approximation algorithm for \problemname.
\end{theorem}

\subsection{Complete k-partite Demands}

In this section, we consider the special case of instances $(P,D)$ where the demands $D$ form a complete k-partite graph. This means, we can partition $P$ into disjoint subsets $S_1,S_2, ..., S_k$ such that $D=S_i\times S_j$ for $i,j \in \{1,...,k\}$ and $i \neq j$. Again, we allow that points share $x$-coordinates but we assume that all points have distinct $y$-coordinates.

\begin{lemma}\label{lem:sparsify-complete-bipartite}
  Let $(P,D)$ be an instance with complete $k$-partite demands and $k$-partition $S_1 \dot{\cup} S_2 \dot{\cup},..., \dot{\cup}S_k$. Then there exists an efficient algorithm that computes a complete $k$-partite instance $(\tilde{P},\tilde{D})$ such that
  \begin{enumerate}
      \item $\tilde{P}\subseteq P$ and $\tilde{D}$ is the complete k-partite demand set on k-partition $S_1\cap \tilde{P}, S_2\cap \tilde{P}, ... ,S_k \cap \tilde{P}$
      \item any feasible solution $Q$ to $(\tilde{P},\tilde{D})$ is feasible for $(P,D)$ as well, and
      \item $|\tilde{P}|=O(\is(\tilde{P},\tilde{D}))$.
  \end{enumerate}
\end{lemma}
\begin{proof}
  An \emph{essential demand} is a pair $(s_i,s_j)\in\tilde{D}$ where $s_i \in S_i$ and $s_j \in S_j$ with $x(s_i)\neq x(s_j)$, such that the vertical side of the rectangle $R(s_i,s_j)$ containing $s_i$ does not contain any other point from $P\setminus S_j$, and such that the vertical side of $R(s_i,s_j)$ containing $s_j$ does not contain any other point from $P\setminus S_i$.

  A point is called essential if there exists some essential demand pair involving this point. A point that is not essential is called \emph{redundant}. Removing a redundant point from the instance is safe in the sense that any feasible solution~$Q$ to the reduced instance is also feasible to the original solution. To see this, let~$s_i\in S_i$ be a redundant point and let~$Q$ be a feasible solution to the reduced instance with point set~$P\setminus\{s_i\}$. Let~$s_j\in S_j$ and assume w.l.o.g.\ that~$x(s_i)<x(s_j)$ and~$y(s_i)<y(s_j)$. Since~$s_i$ is redundant there exists a point~$s\neq s_i, s\in P\setminus S_j$ on~$\lambda(s_i,s_j)$ or a point~$s'\neq s_j, s'\in P\setminus S_i$ on~$\rho(s_i,s_j)$. Consider first the case that there exists a point~$s\neq s_i, s\in P \setminus S_j$ on~$\lambda(s_i,s_j)$. Since~$Q$ is feasible for~$P\setminus\{s_i\}$ there exists a Manhattan path in~$(P\setminus\{s_i\})\cup Q$ from~$s$ to~$s_j$. This path can be extended to a Manhattan path from~$s_i$ to~$s_j$ in~$P\cup Q$ by adding~$s_i$ as a prefix. This establishes feasibility of~$Q$ for~$P$. Now consider the second case that there exists a point~$s'\neq s_j, s'\in P \setminus S_i$ on~$\rho(s_i,s_j)$ but no point~$s\neq s_i, s\in P \setminus S_j$ on~$\lambda(s_i,s_j)$. Let~$s''\in P \setminus S_i$ be the lowest point on~$\rho(s_i,s_j)$. Then~$(s_i,s'')$ is an essential demand contradicting the assumption that~$s_i$ was redundant.

  Our algorithm iteratively removes redundant points until we arrive at a subset $\tilde{P}\subseteq P$ in which every point is essential. Thus for each point $p\in\tilde{P}$ there exists an essential demand pair. Assign the vertical side of the corresponding demand rectangle that contains $p$ as \emph{witness} to $p$.

  The intersection graph of the witnesses is an interval graph. We claim that it has clique size at most 4. Assume to the contrary that there exists some point $q$ (not necessarily a demand point) that is contained in 5 witnesses (each of which has an associated demand point). W.l.o.g.\ assume that at least 3 of the associated demand points lie below $q$. Then there exists a witness (with corresponding demand point $s_i$ and demand pair~$(s_i, s_j)$) that contains two points (possibly including~$s_i$) such that they belong to a different partition than~$s_j \in S_j$. This contradicts the fact that the demand pair~$(s_i, s_j)$ corresponding to this witness is essential. This implies that there exists an independent set of witnesses of size at least $|\tilde{P}|/4$ and thus a boundary independent set of size at least $|\tilde{P}|/8$.
\end{proof}

Combining this lemma with Theorem~\ref{thm: log log n} gives the following result.
\begin{corollary}
There is an $O(\log\log n)$-approximation algorithm for \MGMC/ with complete k-partite demands.
\end{corollary}

\section{Geometric Demands}%
\label{sec:disc_demands}

In this section, we consider restricted demands that are defined implicitly by geometry. We present improved approximation algorithms for these instances.
First, we start with some definition.
\begin{definition}
Let $(P,D)$ be an input to \problemname.
The input demand $D$ is {\em uniform} in $P$ if, for every pair $p,q \in P$, we have a demand $(p,q) \in D$.
\end{definition}

The following theorem was proved in Harmon's thesis~\cite{harmon2006new}.
\begin{theorem}
\label{thm: uniform 2apx}
There is a Greedy algorithm that, on input $(P,D)$ with uniform demand, produces a feasible solution of cost at most $2 \cdot \opt(P,D)$.
\end{theorem}

\begin{definition}
Let $(P,D)$ be an input to \problemname.
The input demand $D$ is called {\em disk demand} if, for each point $p \in P$, there exists a radius $r_p$ such that, $(p,q) \in D$ if and only if~$||p-q||_1 \leq r_p + r_q$.
\end{definition}

In other words, the disk demands are the demands that are defined by an intersection graph of disks.
A special case that is of interest to us is the following:
\begin{definition}
We say that the input demand~$D$ is a \emph{unit-disk demand}, if there exists a real number~$r$ such that~$(p,q) \in D$ if and only if~$|p-q|_1 \leq r$.
\end{definition}

We will use the following lemma in our algorithms to bound the cost.
\begin{lemma}\label{lemma:components}
For any input such that the demand graph~$(P, D)$ is connected,~$\opt(P, D) \geq |P|-1$.
\end{lemma}
\begin{proof}
Let $X$ be the set of points added by an optimal solution. Notice that the Manhattan graph $G_P$ has $|P|$ singleton components. However, the Manhattan graph $G_{X \cup P}$ contains a single component.
Removing one point from any Manhattan graph can only break one component into at most two smaller components; more formally, for any point set $Y \subseteq {\mathbb R}^2$ and $p \in Y$, we have that the number of components in $G_{Y- p}$ and that in $G_{Y}$ differs by at most one.
Therefore, we must have that $|X| \geq |P|-1$.
\end{proof}

\subsection{Algorithms for Unit-Disk Demand}

\newcommand{\grid}{{\mathcal G}}
\newcommand{\gcell}{{\mathcal C}}

Here we present an $O(1)$-approximation algorithm. We assume without loss of generality that $D$ is connected; if not, we could apply the same arguments to each connected component separately.

The algorithm has two steps. First, we draw an arbitrary grid $\grid$ where the space between consecutive horizontal and vertical grid lines is exactly $r/2$; so each grid-cell is an $(r/2)$-by-$(r/2)$ cell. Assume that no input points lie on any cell boundaries.
Let $\gcell$ be a set of grid cells $C$ such that $C$ contains some input point, that is, $C \cap P \neq \emptyset$.
For each grid cell $C \in \gcell$, define the points~$P_C = P \cap C$ inside the cell and~$D_C$ be the induced demands on the points~$P_C$. Observe that for any two points~$p,q$ in the same cell~$C$, we have~$||p-q||_1 \leq r$, so  a demand~$(p,q) \in D_C$. This means that~$(P_C,D_C)$ is in fact an instance with uniform demand.
We make the following simple observations.
\begin{observation}
$\opt(P,D) \geq \sum_{C \in \gcell} \opt(P_C,D_C)$.
\end{observation}

Next, we define our solution points which consist of the {\em inner points}~$S_1$ connecting the demands~$D_C$ in the same cell~$C$, and the {\em outer points}~$S_2$ connecting the demands~$(p,p') \in D$,~$p \in C$ and~$p' \in C'$ where~$C$ and~$C'$ are two different grid cells.
The inner points~$S_1$ are defined as follows. For each cell~$C \in \gcell$, we take the solution~$S_1(C)$ that results from Theorem~\ref{thm: uniform 2apx}, and then~$S_1 = \bigcup_{C \in \gcell} S_1(C)$.
The outer points~$S_2$ is also a union of~$S_2(C)$.
For each grid cell~$C \in \gcell$, for each point~$p \in C$ that has some demand with some other point in different cell, if a vertical grid line~$L$ (a line drawn at the~$x$-coordinate~$x(L)$) is within distance~$r$ from~$p$, we create a point~$(x(L), y(p))$ on line~$L$ (in other words, this is a projection of point~$p$ onto line~$L$.)
We do the same for each point~$p \in P_C$ and each horizontal grid line~$L$: If~$L$ is within distance~$r$ from~$p$, create a projected point~$(x(p), y(L))$. The set~$S_2(C)$ contains, for each~$p \in P_C$, all such projected points from~$p$.

\begin{lemma}[Feasibility]
The solution $S_1 \cup S_2$ is feasible for $(P,D)$.
\end{lemma}
\begin{proof}
Consider any point $(p,q) \in D$. If $p$ and $q$ are in the same cell, we are done since $S_1$ would connect them.
If $p \in C$ and $q \in C'$ for $C,C' \in \gcell$, then consider any grid line $L$ that separates them, so we have that the distance between $L$ and $p$ is at most $r$ and so is the distance between $L$ and $q$.
Let $p'$ and $q'$ be projections of $p$ and $q$ on $L$. We have a Manhattan path $p \rightarrow p' \rightarrow q' \rightarrow q$.
\end{proof}

\begin{lemma}[Cost]
$|S_1| + |S_2| \leq O(\opt(P,D))$
\end{lemma}
\begin{proof}
From Theorem~\ref{thm: uniform 2apx}, we have that $|S_1| \leq 2 \sum_{C \in \gcell} \opt(P_C,D_C) \leq 2 \opt(P,D)$.
Next, for points in $S_2$, notice that each point $p \in P$ creates at most $8$ projected points, so we have:
\[|S_2| = \sum_{C \in \gcell} |S_2(C)| \leq 8 \sum_{C \in \gcell} |P_C| \leq 8 |P| \leq 8 \opt(P,D) + 8 \]
The last inequality follows from Lemma~\ref{lemma:components}.
\end{proof}

\subsection{General Disks}

Now we show that any instance with disc demand admits an $O(\log \Delta)$-approximation where  each~$r_p \in [1,\Delta]$.
As in the proof of the unit disk case, we may assume that $G$ is connected and every input-point in $P$ is part of some demand.

Let $h,k$ be minimal and maximal integers respectively such that $2^h \leq r_p \leq 2^{k-1}$ for all~\mbox{$p\in P$}.
Fix grids $\Gamma_{h} \subseteq \Gamma_{h+1} \subseteq \ldots \subseteq\Gamma_{k}$ of sizes $2^{h-1},2^{h},\ldots,2^{k-1}$ respectively. For every input point, we place 8 points by projecting a point to the two closest grid lines in each direction. We do this for each of the grids adding in total $8n(k-h+1) = O(n \log \Delta)$ many points.

Additionally, we satisfy the demands within a cell of $\Gamma_{h}$ as before by applying the greedy algorithm for each cell separately, using at most $2\cdot\opt$ points. The statement follows, as $\opt \geq \frac{n}{2}$ by Lemma~\ref{lemma:components}.

\subsection{Demands with Two Radii}

In this section, we study the disk demands where there are only two different radii $r_A < r_B$.

\begin{definition}
Let $(P,D)$ be an input to \problemname. The input demand~$D$ is called {\em two-disk demand} if the point set $P$ is partitioned into two subsets $A \cup B$, such that every point $p \in A$ has~$r_p = r_A$; otherwise, $r_p = r_B$ for each $p \in B$.
\end{definition}

Another demand we are interested in is the following:
\begin{definition}
Let $(P,D)$ be an input to \problemname. The input demand~$D$ is called {\em complete bipartite demand} if $P$ is partitioned into two subsets $A \cup B$ such that $(p,q) \in D$ if and only if $p \in A$ and $q \in B$.
\end{definition}

\begin{lemma}
	If there is an~$\alpha$-approximation for the complete bipartite case, then there is an~$O(\alpha)$-approximation for two-disk demands. This implies an $O(\log \log n)$-approximation for two-disk demand.
\end{lemma}
\begin{proof}
Let $(P,D)$ be an input such that $D$ is a two-disk demand without isolated vertices, i.e., every point in~$P$ is part of at least one demand.
Define~$D_A = D\cap (A\times A)$ and~$D_B = D\cap (B\times B)$.
Using the algorithm for unit-disk demands, we have an algorithm satisfying~$D_A$ and~$D_B$ using at most~$O(\opt(A,D_A))$ and~$O(\opt(B,D_B))$ points respectively, which is together $O(\opt(P,D))$.

Consider the grid $\grid$ of size~$r_B/2$ used to satisfy~$D_B$. For each point~$p$ in~$A$ consider the cell~$C \in \gcell$ that contains that point. As before,~$(p,q)\in D$ only for those points $q \in C'$ such that $C$ and $C'$ are in close proximity (that is, constant grid cells away from each other.)
By projecting~$p$ in each direction to the two closest grid lines, we use $O(\opt(P,D))$ points and satisfy all demand pairs between $A$ and $B$ for which both points lie in different grid cells.
Within a cell, the remaining demands are complete bipartite.
Hence, if we have an~$O(\alpha)$-approximation for the complete bipartite case we will have a~$O(\alpha)$-approximation for this case.
\end{proof}

\end{document}